\documentclass[a4paper,11pt]{article}
\pdfoutput=1

\usepackage{graphicx, mathtools, amssymb, latexsym, amsmath, amsfonts, amsthm, bbm, tikz}
\usepackage[toc,page]{appendix}
\usepackage[ruled, linesnumbered, noend]{algorithm2e}
\usepackage{enumitem}

\usepackage{array, fancyhdr, bm, listings, soul, xcolor,titlesec, cancel}
\usepackage[skip=0.5\baselineskip]{caption}

\usepackage{thmtools, thm-restate}
\usepackage{thm-autoref}

\definecolor{Darkblue}{rgb}{0,0,0.4}
\definecolor{Brown}{cmyk}{0,0.61,1.,0.60}
\definecolor{Purple}{cmyk}{0.45,0.86,0,0}
\definecolor{Darkgreen}{rgb}{0.133,0.543,0.133}
\usepackage[colorlinks,linkcolor=Darkblue,filecolor=blue,citecolor=blue,urlcolor=Darkblue,pagebackref]{hyperref}

\theoremstyle{definition}
\newtheorem{theorem}{Theorem}

\newtheorem{lemma}{Lemma}

\title{Matrix Multiplication with Less Arithmetic Complexity and IO Complexity}
\author{Pu Wu\thanks{School Of Computer Science, Peking University. \texttt{puwu1997@126.com}}, Huiqin Jiang\thanks{School Of Computer Science And Cyber Engineering, Guangzhou University. \texttt{hq.jiang@hotmail.com}}, Zehui Shao\thanks{Institute Of Computing Science And Technology, Guangzhou University. \texttt{zshao@gzhu.edu.cn}}, Jin Xu\thanks{Key Laboratory Of High Confidence Software Technologies (Peking University), Ministry Of
Education;
School Of Computer Science,Peking University. \texttt{jxu@pku.edu.cn}}}

\date{}
\usepackage{makecell}
\usepackage{multirow}
\usepackage{multicol}
\usepackage{array}
\usepackage{booktabs}
\usepackage{algorithmic}

\AtBeginDocument{%
  \providecommand\BibTeX{{%
    \normalfont B\kern-0.5em{\scshape i\kern-0.25em b}\kern-0.8em\TeX}}}
\newtheorem{def1}{Definition}
\newtheorem{observation}{Observation}
\newtheorem{cor}{Corollary}
\newtheorem{thm}{Theorem}
\usepackage[ruled,linesnumbered]{algorithm2e} 

\begin{document}
\maketitle

\begin{abstract}
After Strassen presented the first sub-cubic matrix multiplication algorithm,
many Strassen-like algorithms are presented.
Most of them with low asymptotic cost have large hidden leading coefficient which are thus impractical.
To reduce the leading coefficient, Cenk and Hasan give a general approach reducing the leading coefficient of $<2,2,2;7>$-algorithm to $5$ but increasing IO complexity.
In 2017, Karstadt and Schwartz also reduce the leading coefficient of $<2,2,2;7>$-algorithm to $5$ by the Alternative Basis Matrix Multiplication method. Meanwhile, their method reduces the IO complexity and low-order monomials in arithmetic complexity.
In 2019, Beniamini and Schwartz generalize Alternative Basis Matrix Multiplication method reducing leading coefficient in arithmetic complexity but increasing IO complexity.

In this paper, we propose a new matrix multiplication algorithm which reduces leading coefficient both in arithmetic complexity and IO complexity.
We apply our method to Strassen-like algorithms improving arithmetic complexity and IO complexity(the comparison with previous results are shown in Tables 1 and 2).
Surprisingly, our IO complexity of $<3,3,3;23>$-algorithm is $14n^{\log_323}M^{-\frac{1}{2}} + o(n^{\log_323})$ which breaks Ballard's IO complexity low bound(
$\Omega(n^{\log_323}M^{1-\frac{\log_323}{2}})$) for recursive Strassen-like algorithms.

\end{abstract}


\section{Introduction}
Matrix Multiplication is a fundamental computation problem used in many fields.
Strassen\cite{strassen1969gaussian} presented the first non-trivial algorithm with time complexity $O(n^{\log_27})$ which breaks the trivial time complexity $O(n^3)$.
Since then, the matrix multiplication algorithm including the design and analysis of the algorithm has attracted more and more great attention in the last five decades\cite{strassen1969gaussian,2012O,coppersmith1987matrix,1973Duality,le2014powers,karstadt2020matrix,ballard2013graph}.
In those researches, researchers mainly improve and analyze the time complexity in two parts, arithmetic complexity and IO complexity(e.g. the costs of transferring data between the CPU and memory devices, between memory devices and disks or between parallel processors).

For arithmetic complexity, we focus on deriving asymptotic and hidden constants which are improved by reducing the exponent of the arithmetic complexity and the number of additions respectively.
Many Strassen-like algorithms are presented to reduce the exponent of the arithmetic complexity \cite{Scho1981PARTIAL,pan1978strassen,2012O,coppersmith1982asymptotic,coppersmith1987matrix,williams2012multiplying}.
Recently, researchers also use computer-aided techniques\cite{2013On,benson2015framework,smirnov2013bilinear} to discover new matrix multiplication algorithms with less exponent.
But in practice, Srassen-Winograd's algorithm often performs better than some asymptotically faster algorithms\cite{benson2015framework} due to these smaller hidden constants.
This shows the importance of the second research direction in arithmetic complexity, reducing the hidden constants inside the $O$-notation.

For IO complexity, it often costs significant more time than its arithmetic\cite{2004Getting}, which is the reason why we are interested to analyze and reduce the IO complexity.
Clearly, if we run the recursion in the Strassen-like algorithm and put the matrices into the fast memory until the matrices are sufficiently small, we can get an IO complexity of the Strassen-like algorithm  which is $O(({\frac{n}{\sqrt{M}}})^{\log_{n_0}t}M)$\cite{ballard2013graph}.
Furthermore, this bound has been  proved to be tight\cite{ballard2013graph} for Strassen-like algorithms which means the IO complexity can not be improved by changing implementation.

\subsection{Previous Research}

For recursive Strassen-like algorithms, the hidden constant of arithmetic complexity is depended on the number of linear operations in the bilinear function.
So, one way of reducing the hidden constant is to find the bilinear function with less linear operations.
But this way is limited.
Probert proved that 15 additions are necessary for any $<2,2,2;7>$-algorithm\cite{probert1976additive} which means that there is no bilinear function making the hidden constant less than 6 for recursive $<2,2,2;7>$-algorithm.
Surprisingly, this bound can be broken by doing some modification in the Strassen-like algorithm\cite{cenk2017arithmetic,karstadt2020matrix}.

Cenk and Hasan\cite{cenk2017arithmetic} split the Strassen-like algorithm to three linear divided-and-conquer algorithms where two of them transform the inputs into two vectors, followed by vector multiplication of their results, and last of three linear divided-and-conquer algorithms calculates the output.
They reduce the hidden constant in arithmetic complexity of $<2,2,2;7>$-algorithm to 5, whose arithmetic complexity is $5n^{\log_27} + 0.5n^{\log_26}+2n^{\log_25}-6.5n^{2}$ in detail.
Their method can also apply in other Strassen-like algorithms, such as $<2,3,4;20>$-algorithm, $<3,3,3;23>$-algorithm, and $<6,3,3;40>$-algorithm.
However, it increases the IO cost and memory footprint, whose IO complexity of $<2,2,2;7>$-algorithm is $O(n^{\log_27})$ in detail.

Karstadt and Schwartz\cite{karstadt2020matrix} present the Alternative Basis Matrix Multiplication method which uses the basis transformation to pre-compute, followed by applying recursive Strassen-like algorithm on their results, and uses the basis transformation to calculate the output.
They reduce the hidden constants both in arithmetic complexity and IO complexity of $<2,2,2;7>$-algorithm from 6,5 to 5,4 respectively, whose arithmetic complexity and IO complexity are $5n^{\log_27} - 4n^{2} + 3n^2\log_2n$\cite{karstadt2020matrix} and $4(\frac{\sqrt{3}n}{\sqrt{M}})^{\log_27}M - 12n^2 + 3n^2\log_2(\sqrt{2}\frac{n}{\sqrt{M}}) + 5M$\cite{karstadt2020matrix} respectively in detail.

Beniamini and Schwartz\cite{beniamini2019faster} present the Sparse Decomposition method by generalizing the Alternative Basis Matrix Multiplication method with using large basis.
It gets less arithmetic complexity than Alternative Basis Matrix Multiplication method in some Strassen-like algorithms, for example, the arithmetic complexity of $<3,3,3;23>$-algorithm obtained by Sparse Decomposition is $2n^{\log_323} + 3n^{\log_320} + 2n^{\log_314} + 2n^{\log_312} + 2n^{\log_311} + 33n^{\log_310} -43n^2$\cite{beniamini2019faster} which is less than $6.57n^{\log_323} + \frac{20}{9}n^2\log_3n - 5.57n^2$ obtained by Alternative Basis Matrix Multiplication\cite{karstadt2020matrix}.
However, it increases the IO complexity and memory footprint of $<3,3,3;23>$-algorithm, whose IO complexity is $12.64n^{\log_323}M^{1-\log_{20}23} + O(n^{\log_320})$\cite{karstadt2020matrix,beniamini2019faster} in detail.

\subsection{Our Contribution}
We present a new method called Algebra Decomposition method which improves both arithmetic complexity and IO complexity in some Strassen-like algorithms.
For example, we improve the arithmetic complexity and IO complexity of $<3,3,3;23>$-algorithm to $2n^{\log_323} + 4.56n^{\log_{27}(23^3-4)} - 5.56n^2 + 0.77n^2\log_3n$ and
$14n^{\log_323}M^{-0.5} - $ $ 6n^{\frac{\log_{27}(23^3-4)}{3}}M^{-0.5} + $ $2M - 16.69n^2$ $ +2.33n^2\log_{3}{\sqrt{2}\frac{n}{
\sqrt{M}}}+ 28.07n^{\frac{\log_{27}{23^3-0.77}}{3}}M^{-0.42}$
respectively.
Notice that our IO complexity seemingly contradicts Ballard's lower bound(Theorem \ref{cite1})\cite{ballard2013graph}.
But actually, Ballard's lower bound is based on recursive Strassen-like algorithm, and our algorithm is obtained by doing modification on Strassen-like algorithm.

\begin{table}[t]\tiny\centering
  \caption{$<3,3,3;23>$-algorithms}\centering
  \label{tab:freq}
  \begin{tabular}{m{2.4cm}<{\centering}|m{1.3cm}<{\centering}|m{9cm}<{\centering}}
    \midrule
   \textbf{ {\fontsize{6.5pt}{0pt}Algorithm}} & \textbf{{\fontsize{6.5pt}{0pt}Complexity} }& \textbf{{\fontsize{6.5pt}{0pt} Results}} \\ \hline
    \rule{0pt}{16pt}
    \multirow{2}{*}{Original\cite{benson2015framework}} & Arithmetic& $7.93n^{\log_323}-6.93n^2$\\ \cline{2-3}
    \rule{0pt}{16pt}
    & IO&$47.62n^{\log_323}M^{-0.42}-20.79n^2$\\\hline
    \rule{0pt}{16pt}
    \multirow{2}{*}{Cenk-Hasan\cite{cenk2017arithmetic}} & Arithmetic & $2n^{\log_323} + 6.75n^{\log_321} - 7.75n^2$ \\ \cline{2-3}
     \rule{0pt}{16pt}
    & IO& $n^{\log_323} - 1.85n^{\log_320}M^{0.02} + 10.42n^{\log_{3}21}M^{-0.27} - 7.75n^2$\\ \hline
    \rule{0pt}{8pt}
    \multirow{2}{*}{Karstadt-Schwartz\cite{karstadt2020matrix}} & Arithmetic  & $6.58n^{\log_323} + 0.33n^2\log_3n - 5.58n^2$\\\cline{2-3}
    \rule{0pt}{16pt}
    & IO& $ 31.5n^{\log_323}M^{-0.42} - 14.71n^2 + 2n^2\log_2{\sqrt{2}\frac{n}{\sqrt{M}}}+2M$\\ \hline
    \rule{0pt}{13pt}
    \multirow{2}{*}{Beniamini-Schwartz\cite{beniamini2019faster}} & Arithmetic &$2n^{\log_323} + 3n^{\log_320} + 2n^{\log_314} +  2n^{\log_312} + 2n^{\log_311} + 33n^{\log_310} -43n^2$ \\ \cline{2-3}
     \rule{0pt}{25pt}
    & IO& \makecell{$6.32n^{log_{3}23}M^{-0.04} + n^{log_{3}20}(13.03M^{-0.11} - 1)+$\\$n^{log_{3}14}(18.75M^{-0.05}-10) - 129n^2  + n^{\log_{3}12}(24.59M^{-0.03}-16) +$\\$
n^{\log_{3}11}(30.84M^{-0.03}-22) + n^{\log_{3}10}(133.16M^{-0.04}-29)$}\\ \hline
     \rule{0pt}{16pt}
    \multirow{2}{*}{\textbf{Algorithm \ref{fin1}(Ours)}} & Arithmetic& $2n^{\log_323} + 4.56n^{\log_{27}(23^3-4)} - 5.56n^2 + 0.33n^2\log_3n$\\ \cline{2-3}
    \rule{0pt}{30pt}
    & IO& \makecell{$14n^{\log_323}M^{-0.5} - 6n^{\frac{\log_{27}(23^3-4)}{3}}M^{-0.5} + $\\$ 28.07n^{\frac{\log_{27}{23^3-0.77}}{3}}M^{-0.42} - 16.69n^2 + 2n^2\log_{3}{\sqrt{2}\frac{n}{
\sqrt{M}}}+2M$}\\
  \bottomrule
\end{tabular}
\end{table}

\begin{table}[t]\tiny
	\centering
	\caption{Algebra Decomposition Algorithms}
	\begin{tabular}{m{1.5cm}<{\centering}|m{1.4cm}<{\centering}|m{2.8cm}<{\centering}|m{0.95cm}<{\centering}|m{0.95cm}<{\centering}|m{3.9cm}<{\centering}}
		\hline
        \rule{0pt}{8pt}
        \multirow{2}{*}{\textbf{ {\fontsize{6.5pt}{0pt}Algorithm}}} & \multirow{2}{*}{\textbf{ {\fontsize{6.5pt}{0pt}Complexity}}} & \multirow{2}{*}{\textbf{{\fontsize{6.5pt}{0pt}Leading Monomial}}} & \multicolumn{3}{c}{\textbf{{\fontsize{6.5pt}{0pt}Leading Coefficient}}} \\ \cline{4-6}
        \rule{0pt}{8pt}
        & & &\textbf{Original} & \textbf{Previous}& \textbf{Algorithm \ref{fin2}(Ours)}\\ \hline
        \rule{0pt}{8pt}
        \multirow{2}{*}{$<3,2,3;15>$} &  Arithmetic & $n^{3\log_{18}15}$ &15.06\cite{benson2015framework} & 7.94\cite{karstadt2020matrix} & $5.62+1.73M^{-0.01}-3.23M^{-0.405}$ \\ \cline{2-6}
        \rule{0pt}{16pt}
        & IO & $n^{3\log_{18}15}M^{1-\frac{3\log_{18}15}{2}}$ & 70.52\cite{benson2015framework}& 37.19\cite{karstadt2020matrix} & 32.04\\ \hline
        \rule{0pt}{8pt}
         \multirow{2}{*}{$<2,3,4;20>$} &  Arithmetic & $n^{3\log_{24}20}$ &9.96\cite{benson2015framework} & 7.46\cite{karstadt2020matrix} & $3.66+3.37M^{-0.011}-7.88M^{-0.413} $ \\ \cline{2-6}
        \rule{0pt}{16pt}
        & IO & $n^{3\log_{24}20}M^{1-\frac{3\log_{24}20}{2}}$ & 47.08\cite{benson2015framework} &35.27\cite{karstadt2020matrix} &26.32 \\ \hline

        \rule{0pt}{8pt}
        \multirow{2}{*}{$<6,3,3;40>$} &  Arithmetic & $n^{3\log_{54}40}$ &55.63\cite{smirnov2013bilinear} & 9.39\cite{karstadt2020matrix} & $6.42+2.68M^{-0.01}-5.08M^{-0.387}$ \\ \cline{2-6}
        \rule{0pt}{16pt}
        & IO & $n^{3\log_{54}40}M^{1-\frac{3\log_{54}40}{2}}$ & 255.35\cite{smirnov2013bilinear}&43.11\cite{karstadt2020matrix} &36.58  \\ \hline
	\end{tabular}
\end{table}

\begin{theorem}\label{cite1}\cite{ballard2013graph}
The IO complexity $IO(n)$ of a recursive Strassen-like fast matrix multiplication algorithm with $O(n^{\omega_0})$ arithmetic operations, on a machine with fast memory of size $M$ is
$$IO(n) = \Omega((\frac{n}{\sqrt{M}})^{\omega_0}M)$$
\end{theorem}

Comparing results of Alternative Basis Matrix Multiplication method, our result improves both arithmetic complexity and IO complexity. And our result improves IO complexity but increases arithmetic complexity in low-order monomials comparing with Sparse Decomposition method.
Based on our main ideal, we will present two algorithms where Algorithm \ref{fin1} improves leading coefficient both in arithmetic complexity and IO complexity and Algorithm \ref{alg:W6} is better than Algorithm \ref{fin1} in some cases but worst in other cases.
We show the results of them in Tables 1 and 2.

\subsection{Organization}

In Section 2, we will show some useful algebra results which are mathematical foundations of our algorithms.
In Section 3, we will describe our main algorithms, Algorithm \ref{fin1} and Algorithm \ref{alg:W6}.
In Section 4, we will analyze their complexity. Specifically, we present the arithmetic complexity in Section 4.2, and IO complexity in Section 4.3.
In Section 5, we give an example, $<3,3,3;23>$-algorithm, of Algorithms \ref{fin1},\ref{alg:W6} and we also show the format of Appendix.
In Appendix, we give the decompositions of the fast matrix multiplication algorithms showed in the Tables 1 and 2.

\section{Preliminary}
Let $R$ be a ring.
Once we define a linear map $\varphi: R^{p_0\times q_0}\rightarrow R^{p_1\times q_1}$, we correspondingly define $\varphi(A):R^{r_1p_0\times r_2q_0}\rightarrow R^{r_1p_1\times r_2q_1}$ for any $A\in R^{r_1p_0\times r_2q_0}$ as $\varphi(A) := \varphi(B)$ where $B = (b_{ij})_{p_0\times q_0}$ and $b_{ij}$ is the i-th, j-th $r_1\times r_2$ size submatrix of $A$.

\begin{def1}
Let $a = (a_i)_{1\times p}, b = (b_i)_{1\times q}$.
Define $a\bigoplus b = (a_1,a_2,\ldots,a_p,b_1,b_2,\ldots,b_q)$, $\bigoplus\limits_{i=1}^tc_i = c_1\bigoplus(\bigoplus\limits_{i=2}^tc_i)$.
\end{def1}

\begin{def1}
Let $\varphi: R^{1\times p_0}\rightarrow R^{1\times q_0}$ be a linear map. We recursively define a linear map $\varphi^{k+1}: R^{1\times p}\rightarrow R^{1\times q}$ (where $p = p_0^{k+1}, q = q_0^{k+1}$) by $\varphi^{k+1}(A) = \varphi(\varphi^{k}(A_{1,1}),\ldots,\varphi^{k}(A_{1,p_0}))$, where $A = (A_{1,1},\ldots,A_{1,p_0})$ and $A_{1,j}$ are $1\times \frac{p}{p_0}$ subvectors.
\end{def1}

\begin{def1}
For a linear map $\ell: \ell(a) = (a_{i_1},a_{i_2},\ldots,a_{i_t})$ where $a = (a_1,a_2,\ldots,a_k), 1\leq i_1<i_2<\ldots<i_t\leq k$, we call $\ell$ a interception map. Denote $\ell^j(a) = a_{i_j}$,
identity map($I(a) = a$) as $I$.
\end{def1}

\begin{def1}
Let $\varphi: R^{1\times p_1}\times R^{1\times p_2}\times\ldots \times R^{1\times p_k}\rightarrow R^{1\times q}$ be a linear map, $a_i=(a_{i,j})_{1\times rp_i}$ for $1\leq i\leq k$,
$a_{i_1,i_2,\ldots,i_t} = (a_{i_1,i_2,\ldots,i_t,j})_{1\times \prod\limits_{j=1}^tp_{i_j}}$ for $1\leq i_1\leq k,\ldots,1\leq i_t\leq k$. \\
Denote
$${\sum\limits_{i=1}^k}^{\varphi}a_{i} = (\varphi(b_{1,1},\ldots,b_{k,1}),\varphi(b_{1,2}, \ldots,b_{k,2}), \ldots,\varphi(b_{1,r},\ldots,b_{k,r}))$$
where $b_{i,j} = (a_{i,(j-1)*p_i+1},$ $a_{i,(j-1)*p_i+2},$ $\ldots,a_{i,(j-1)*p_i+p_i})$,
$${\sum\limits_{i_1=1}^k}^{\varphi}{\sum\limits_{i_2=1}^k}^{\varphi}\ldots{\sum\limits_{i_t=1}^k}^{\varphi}a_{i_1,i_2,\ldots,i_t} = {\sum\limits_{i_1=1}^k}^{\varphi}({\sum\limits_{i_2=1}^k}^{\varphi}(\ldots{\sum\limits_{i_t=1}^k}^{\varphi}a_{i_1,i_2,\ldots,i_t}))\ldots).$$
\end{def1}

\begin{observation}\label{oo01}
Let $\varphi: R^{1\times p_1}\times R^{1\times p_2}\times\ldots \times R^{1\times p_k}\rightarrow R^{1\times q}$ be a linear map, $a_{i,j}=(a_{i,j,z})_{1\times r_jp_i}, b_{i,j} = (b_{i,j,z})_{1\times rp_i}$ and $\lambda_j$ be real number where $1\leq i\leq k, 1\leq j\leq t$.
Then,
$$ {\sum\limits_{i=1}^k}^{\varphi}\bigoplus\limits_{j=1}^ta_{i,j} = \bigoplus\limits_{j=1}^t{\sum\limits_{i=1}^k}^{\varphi}a_{i,j} \text{\,\,\,\,and\,\,\,\,}  {\sum\limits_{i=1}^k}^{\varphi}\sum\limits_{j=1}^t\lambda_j b_{i,j} = \sum\limits_{j=1}^t\lambda_j{\sum\limits_{i=1}^k}^{\varphi}b_{i,j}.$$
\end{observation}

\begin{def1}\label{def1}
Let $a=(a_i)_{1\times p^t}$, $\varphi_i: R^{1\times p_i}\rightarrow R^{1\times q_i}$ be a linear map and $\ell_i: R^{p}\rightarrow R^{p_i}$ be an interception map where $1\leq i\leq t$.
Denote $$\varphi_1\ell_1\circ\varphi_2\ell_2\ldots\circ\varphi_t\ell_t(a) = \varphi_1(\bigoplus\limits_{i=1}^{p_1}\varphi_2\ell_2\circ\ldots\circ\varphi_t\ell_t(\ell_1^i(a')))$$
where $a' = (b_j)_{1\times p}$ and $b_j$ is $1\times p^{t-1}$ subvector of $a$.
Denote $\varphi I$ as $\varphi$, $I\ell$ as $\ell$ and $\varphi_{1}\circ\varphi_{2}\circ\ldots\circ\varphi_{t}(A)$
as ${\prod\limits_{j=1}^t}^{\circ}\varphi_{j}(A)$ shortly.
\end{def1}
\begin{observation}\label{oo91}
Following the definition \ref{def1}, we have
$$\varphi_1\ell_1\circ\varphi_2\ell_2\ldots\circ\varphi_t\ell_t(a) = \varphi_{1}\circ\varphi_{2}\ldots\circ\varphi_{k}(\ell_{1}\circ\ell_{2}\circ\ldots \circ\ell_{t}(a)).$$
\end{observation}

\begin{lemma}\label{l1}
Let $a = (a_i)_{1\times p}, b = (b_i)_{1\times p^k}$, $\psi: R^{1\times p}\rightarrow R^{n\times m}$, $\varphi_i: R^{1\times p_i}\rightarrow R^{1\times q_i}$, $\varphi: R^{1\times p_1}\times R^{1\times p_2}\times\ldots\times R^{1\times p_t}\rightarrow R^{n\times m}$ be linear maps and $\ell_i: R^{1\times p}\rightarrow R^{1\times p_i}$ be an interception map, where $1\leq i\leq t$.
If $\psi(a) = \varphi(\varphi_1\ell_1(a),\varphi_2\ell_2(a),\ldots,\varphi_t\ell_t(a))$, then
\begin{equation*}
\begin{split}
\psi^k(b) = & {\sum\limits_{i_1=1}^t}^{\varphi}{\sum\limits_{i_2=1}^t}^{\varphi}\ldots{\sum\limits_{i_k=1}^t}^{\varphi}{\prod\limits_{j=1}^k}^{\circ}\varphi_{i_j}\ell_{i_j}(b).\\
\end{split}
\end{equation*}
\end{lemma}
\begin{proof}
We will prove this by induction on $t$.
First, this lemma holds when $t=1$.
Assume that it holds on $t-1$.

Let $c_i = (b_{(i-1)*p^{k-1}+j})_{1\times p^{k-1}}, c = (c_i)_{1\times p}, d_i = \psi^{k-1}(c_i), e_i = {\sum\limits_{i_2=1}^t}^{\varphi}\ldots{\sum\limits_{i_k=1}^t}^{\varphi}\varphi_{i_2}\ell_{i_2}\circ\ldots\circ\varphi_{i_k}\ell_{i_k}(c_i)$.
W.l.o.g, we assume that $\varphi_{i_1}\ell_{i_1}(A) = \bigoplus\limits_{j=1}^{q_{i_1}}\sum\limits_{i=1}^{p}\lambda_{j,i}^{i_1}A_{i}$ where $A = (A_i)_{1\times p}$.
Therefore,
\begin{equation*}
\begin{aligned}
&{\sum\limits_{i_1=1}^t}^{\varphi}{\sum\limits_{i_2=1}^t}^{\varphi}\ldots{\sum\limits_{i_k=1}^t}^{\varphi}{\prod\limits_{j=1}^{k}}^{\circ}\varphi_{i_j}\ell_{i_j}(b)&\\ =& {\sum\limits_{i_1=1}^t}^{\varphi}{\sum\limits_{i_2=1}^t}^{\varphi}\ldots{\sum\limits_{i_k=1}^t}^{\varphi}\varphi_{i_1}(\bigoplus\limits_{z=1}^{p_{i_1}}{\prod\limits_{j=2}^{k}}^{\circ}\varphi_{i_j}\ell_{i_j}(\ell_{i_1}^z(c)))\\
=&{\sum\limits_{i_1=1}^t}^{\varphi}{\sum\limits_{i_2=1}^t}^{\varphi}\ldots{\sum\limits_{i_k=1}^t}^{\varphi}\varphi_{i_1}
\ell_{i_1}({\prod\limits_{j=2}^{k}}^{\circ}\varphi_{i_j}\ell_{i_j}(c_1),{\prod\limits_{j=2}^{k}}^{\circ}\varphi_{i_j}\ell_{i_j}(c_2),\ldots,{\prod\limits_{j=2}^{k}}^{\circ}\varphi_{i_j}\ell_{i_j}(c_p))&\\
=&{\sum\limits_{i_1=1}^t}^{\varphi}{\sum\limits_{i_2=1}^t}^{\varphi}\ldots{\sum\limits_{i_k=1}^t}^{\varphi}
\bigoplus\limits_{j=1}^{q_{i_1}}\sum\limits_{i=1}^{p}\lambda_{j,i}^{i_1}{\prod\limits_{j=2}^{k}}^{\circ}\varphi_{i_j}\ell_{i_j}(c_i).&\\
\end{aligned}
\end{equation*}
By Observation \ref{oo01},
\begin{equation*}
\begin{aligned}
&{\sum\limits_{i_1=1}^t}^{\varphi}{\sum\limits_{i_2=1}^t}^{\varphi}\ldots{\sum\limits_{i_k=1}^t}^{\varphi}
\bigoplus\limits_{j=1}^{q_{i_1}}\sum\limits_{i=1}^{p}\lambda_{j,i}^{i_1}{\prod\limits_{j=2}^{k}}^{\circ}\varphi_{i_j}\ell_{i_j}(c_i)&\\
=&{\sum\limits_{i_1=1}^t}^{\varphi}\bigoplus\limits_{j=1}^{q_{i_1}}\sum\limits_{i=1}^{p}\lambda_{j,i}^{i_1}{\sum\limits_{i_2=1}^t}^{\varphi}\ldots{\sum\limits_{i_k=1}^t}^{\varphi}
{\prod\limits_{j=2}^{k}}^{\circ}\varphi_{i_j}\ell_{i_j}(c_i) \\
=&{\sum\limits_{i_1=1}^t}^{\varphi}\bigoplus\limits_{j=1}^{q_{i_1}}\sum\limits_{i=1}^{p}\lambda_{j,i}^{i_1}e_i.&\\
\end{aligned}
\end{equation*}

By the induction hypothesis, we have $d_i = e_i$. Then, we get
\begin{equation*}
\begin{split}
&{\sum\limits_{i_1=1}^t}^{\varphi}{\sum\limits_{i_2=1}^t}^{\varphi}\ldots{\sum\limits_{i_k=1}^t}^{\varphi}{\prod\limits_{j=1}^{k}}^{\circ}\varphi_{i_j}\ell_{i_j}(b)\\
=&{\sum\limits_{i_1=1}^t}^{\varphi}\bigoplus\limits_{j=1}^{q_{i_1}}\sum\limits_{i=1}^{p}\lambda_{j,i}^{i_1}e_i\\
=&{\sum\limits_{i_1=1}^t}^{\varphi}\bigoplus\limits_{j=1}^{q_{i_1}}\sum\limits_{i=1}^{p}\lambda_{j,i}^{i_1}d_i\\
=&{\sum\limits_{i_1=1}^t}^{\varphi}\varphi_{i_1}\ell_{i_1}((d_1,d_2,...,d_p))\\
=&\psi^k(b).
\end{split}
\end{equation*}
\end{proof}

\begin{def1}
Let $\varphi_i: R^{m_i\times n_i}\rightarrow R^{1\times q_i}$ be a linear map where $1 \leq i \leq t$,
$A \in R^{r_1\prod\limits_{k=1}^{t}m_{k}\times r_2\prod\limits_{k=1}^{t}n_{k}}$
and $A_{i,j}$ be $r_1\prod\limits_{k=1}^{t-1}m_{k}\times r_2\prod\limits_{k=1}^{t-1}n_{k}$ submatrix of $A$.
Define ${\prod\limits_{j=1}^t}^\ast\varphi_{j}(A) = \varphi_{1}\ast\varphi_{2}\ast\ldots\ast\varphi_{t}(A) = \bigoplus\limits_{i=1}^{q_t}\varphi_{1}\ast\varphi_{2}\ast\ldots\ast\varphi_{t-1}(b_i)$
where $(b_i)_{1\times q_t} = \varphi_{t}(A)$.
\end{def1}

\begin{observation}
$\varphi_{1}\ast\varphi_{2}\ast\ldots\ast\varphi_{t}(\sum_{i=1}^{t}\lambda_iA_i) = \sum_{i=1}^t\lambda_i\varphi_{1}\ast\varphi_{2}\ast\ldots\ast\varphi_{t}(A_i).$
\end{observation}

\begin{lemma}\label{l2}
Let $\ell_i: R^{1\times u_i}\rightarrow R^{1\times v_i}$ be an interception map, $\varphi_i: R^{p_i\times q_i}\rightarrow R^{1\times u_i}$ be a linear map where $1\leq i\leq k$ and $A\in R^{r_1\prod\limits_{i=1}^kp_i\times r_2\prod\limits_{i=1}^kq_i}$.
We have
$$\ell_1\circ \ell_2\circ\ldots\circ \ell_k(\varphi_{k}\ast \varphi_{k-1}\ast\ldots\ast \varphi_{1}(A)) = \ell_k\varphi_{k}\ast \ell_{k-1}\varphi_{k-1}\ast\ldots\ast \ell_1\varphi_{1}(A).$$
\end{lemma}
\begin{proof}
We prove this by induction on $k$. First, this lemma holds when $k=1$.
Assume this lemma holds on $k-1$.
W.l.o.g, we assume $\varphi_1(A) = (b_1, b_2, ..., b_{u_1})$.
Thus,
\begin{equation*}
\begin{split}
{\prod\limits_{i=1}^k}^{\circ}\ell_i({\prod\limits_{j=k}^{1}}^{\ast}\varphi_{j}(A)) =&{\prod_{i=1}^k}^{\circ}\ell_i(\bigoplus\limits_{t=1}^{u_1}{\prod\limits_{j=k}^{2}}^{\ast}\varphi_{j}(b_t))\\
=&\bigoplus\limits_{h=1}^{v_1}{\prod_{i=2}^k}^{\circ}\ell_i(\ell_1^h(\bigoplus\limits_{t=1}^{u_1}({\prod\limits_{j=k}^{2}}^{\ast}\varphi_{j}(b_t))))\\
=&\bigoplus\limits_{h=1}^{v_1}{\prod_{i=2}^k}^{\circ}\ell_i({\prod\limits_{j=k}^{2}}^{\ast}\varphi_{j}(\ell_1^h(b_1,b_2,...,b_{u_1}))).\\
\end{split}
\end{equation*}
By the induction hypothesis, we have
\begin{equation*}
\begin{split}
{\prod\limits_{i=1}^k}^{\circ}\ell_i({\prod\limits_{j=k}^{1}}^{\ast}\varphi_{j}(A)) = & \bigoplus\limits_{h=1}^{v_1}{\prod_{i=2}^k}^{\circ}\ell_i({\prod\limits_{j=k}^{2}}^{\ast}\varphi_{j}(\ell_1^h(b_1,b_2,...,b_{u_1})))\\
= & \bigoplus\limits_{h=1}^{v_1}{\prod\limits_{j=k}^{2}}^{\ast}\ell_j\varphi_{j}(\ell_1^h(b_1,b_2,...,b_{u_1}))\\
= & \bigoplus\limits_{h=1}^{v_1}{\prod\limits_{j=k}^{2}}^{\ast}\ell_j\varphi_{j}(\ell_1^h(\varphi_1(A)))\\
= & {\prod\limits_{j=k}^{1}}^{\ast}\ell_j\varphi_{j}(A).\\
\end{split}
\end{equation*}
\end{proof}
\begin{lemma}\label{l3}
Let $\varphi_i: R^{m_i\times n_i}\rightarrow R^{1\times q_i}$, $\phi_i:R^{m_i\times n_i}\rightarrow R^{1\times p_i}$ and $\psi_i:R^{1\times p_i}\rightarrow R^{1\times q_i}$ be linear maps for $1 \leq i \leq k$,
$A\in R^{r_1\prod\limits_{i=1}^{k}m_{i}\times r_2\prod\limits_{i=1}^{k}n_{i}}$.
If $\psi_i\phi_i = \varphi_i$ for $1\leq i\leq k$, then
$$\varphi_{1}\ast\varphi_{2}\ast\ldots\ast\varphi_{k}(A) = \psi_{1}\ast\psi_{2}\ast\ldots\ast\psi_{k}(\phi_{1}\ast\phi_{2}\ast\ldots\ast\phi_{k}(A)).$$
\end{lemma}
\begin{proof}
We prove this by induction on $k$. First, this lemma holds when $k=1$.
Assume this lemma holds on $k-1$. 
W.l.o.g, we assume $\varphi_{k}(A) = \bigoplus\limits_{r=1}^{q_{k}}\sum\limits_{i,j}\lambda_{ij}^{r}A_{ij}$ where $A = (A_{ij})_{m_{k}\times n_{k}}$, $\psi_{k}(B) = \bigoplus\limits_{r=1}^{q_k}\sum\limits_{z=1}^{p_k}\nu_{z}^{r}B_{z}$ where $B = (B_{z})_{1\times p_{k}}$, and $\phi_{k}(A) = \bigoplus\limits_{z=1}^{q_k}\sum\limits_{i,j}\mu_{ij}^{z}A_{ij}$ where $A = (A_{ij})_{m_{k}\times n_{k}}$.
Since $\varphi_{k} = \psi_{k}\phi_{k}$, we have
$\sum\limits_{i,j}\lambda_{ij}^{r}A_{ij} = \sum\limits_{z=1}^{p_{k}}\nu_{z}^{r}\sum\limits_{i,j}\mu_{ij}^{z}A_{ij}$.
By the induction hypothesis, we have
\begin{equation*}
\begin{split}
\varphi_{1}\ast\varphi_{2}\ast\ldots\ast\varphi_{k}(A)
&= \bigoplus\limits_{r=1}^{q_i}\varphi_{1}\ast\varphi_{2}\ast\ldots\ast\varphi_{{k-1}}(\sum_{i,j}\lambda_{ij}^rA_{ij})\\
& = \bigoplus\limits_{r=1}^{q_i}\psi_{1}\ast\psi_{2}\ast\ldots\ast\psi_{{k-1}}(\phi_{1}\ast\phi_{2}\ast\ldots\ast\phi_{k-1}(\sum_{z=1}^{p_{k}}\nu_{z}^r\sum_{i,j}\mu_{ij}^zA_{ij}))\\
& = \bigoplus\limits_{r=1}^{q_i}\psi_{1}\ast\psi_{2}\ast\ldots\ast\psi_{k-1}(\sum_{z=1}^{p_{k}}\nu_{z}^r\phi_{1}\ast\phi_{2}\ast\ldots\ast\phi_{{k-1}}(\sum_{i,j}\mu_{ij}^zA_{ij}))\\
& = \psi_{1}\ast\psi_{2}\ast\ldots\ast\psi_{{k}}(\bigoplus\limits_{z=1}^{p_{k}}\phi_{1}\ast\phi_{2}\ast\ldots\ast\phi_{{k-1}}(\sum_{i,j}\mu_{ij}^zA_{ij}))\\
& = \psi_{1}\ast\psi_{2}\ast\ldots\ast\psi_{{k}}(\phi_{1}\ast\phi_{2}\ast\ldots\ast\phi_{{k}}(A)).\\
\end{split}
\end{equation*}
\end{proof}

\section{The Design Of Algorithm}
We refer to the $<U_{t\times n_0m_0},V_{t\times m_0k_0},W_{t\times n_0k_0}>$ of a $<n_0,m_0,k_0;t>$-algorithm as its encoding/decoding matrices\cite{benson2015framework} (where $U,V$ are encoding matrices and $W$ is the decoding matrix).
An encoding/decoding matrix is corresponding to a bilinear function $f: R^{1\times n_0m_0}\times R^{1\times m_0k_0}\rightarrow R^{1\times n_0k_0}$ with
$$f(x,y) = W^{T}((U\cdot x)\odot(V\cdot y))$$
where $x\in R^{1\times n_0m_0}, y\in R^{1\times m_0k_0}$ and $\odot$ is element-wise vector product(Hadamard product).
Let $\psi(A) = U\cdot x, \phi(B) = V\cdot y,\varphi(C) = W^{T}C$ where $A = (a_{ij})_{n_0\times m_0}, B = (b_{ij})_{m_0\times k_0}, C = (c_{ij})_{1\times t}, x = (a_{ij})_{1\times n_0m_0}, y = (b_{ij})_{1\times m_0k_0}$.

Now, we will decompose the linear maps $\varphi,\psi,\phi$ to get the faster algorithm.
Let $\varphi_i: R^{1\times p_i}\rightarrow R^{1\times q_i},\psi_i: R^{n_0\times m_0}\rightarrow R^{1\times p_i}, \psi_{i,1}: R^{n_0\times m_0}\rightarrow R^{1\times p_{i,1}}, \psi_{i,2}: R^{1\times p_{i,1}}\rightarrow R^{1\times p_{i}}, \phi_i:R^{m_0\times k_0}\rightarrow R^{1\times p_i}, \phi_{i,1}: R^{m_0\times k_0}\rightarrow R^{1\times p'_{i,1}}, \phi_{i,2}: R^{1\times p'_{i,1}}\rightarrow R^{1\times p_{i}}$ be linear maps and $\ell_{i}$ be interception maps which satisfy that $\varphi = {\sum\limits_{i=1}^{h}}^{\varphi_0}\varphi_i\ell_{i}$,
$\psi_i = \ell_{i}\psi$, $\phi_i = \ell_{i}\phi$
and $\psi_i = \psi_{i,2}\psi_{i,1}, \phi_i = \phi_{i,2}\phi_{i,1}$, where $1\leq i\leq h$.

In this paper, we consider the situation that $\varphi_i = \psi_{i,2} = \phi_{i,2}= I ,p_i= q_i = p_{i,1} = p'_{i,1} = 1$ for $2\leq i\leq h$, $p_{1,1}\geq 2, p'_{1,1}\geq 2, q_1\geq 2, t-p_1>\max\{m_0n_0,m_0k_0,n_0k_0\}$ and $h = t - p_1+1$.

\bigskip
\noindent

Let $A\in R^{n_0^q\times m_0^q}, B\in R^{m_0^q\times k_0^q}$ be the input of $<m_0,n_0,k_0;t>$-algorithm and $C$ be the output.
Then, we have\cite{cenk2017arithmetic,karstadt2020matrix,beniamini2019faster}
$$C = \varphi^{q}({\prod\limits_{i=1}^q}^{\ast}\psi(A)\odot{\prod\limits_{i=1}^q}^\ast\phi(B)).$$

Further more, by Lemma \ref{l1}, and Observation \ref{oo91}, we have
\begin{equation*}
\begin{split}
C = & {\sum\limits_{i_1=1}^{h}}^{\varphi_0}{\sum\limits_{i_2=1}^{h}}^{\varphi_0}
\ldots{\sum\limits_{i_q=1}^{h}}^{\varphi_0}{\prod\limits_{j=1}^q}^{\circ }\varphi_{i_j}\ell_{i_j}
({\prod\limits_{i=1}^q}^{\ast}\psi(A)
\odot{\prod\limits_{i=1}^q}^{\ast}\phi(B))\\
= & {\sum\limits_{i_1=1}^{h}}^{\varphi_0}{\sum\limits_{i_2=1}^{h}}^{\varphi_0}
\ldots{\sum\limits_{i_q=1}^{h}}^{\varphi_0}{\prod\limits_{j=1}^q}^{\circ }\varphi_{i_j}
({\prod\limits_{z=1}^q}^{\circ }\ell_{i_z}({\prod\limits_{i=1}^q}^{\ast}\psi(A)
\odot{\prod\limits_{i=1}^q}^{\ast}\phi(B)))\\
= & {\sum\limits_{i_1=1}^{h}}^{\varphi_0}{\sum\limits_{i_2=1}^{h}}^{\varphi_0}
\ldots{\sum\limits_{i_q=1}^{h}}^{\varphi_0}{\prod\limits_{j=1}^q}^{\circ }\varphi_{i_j}
({\prod\limits_{z=1}^q}^{\circ }\ell_{i_z}({\prod\limits_{i=1}^q}^{\ast}\psi(A))
\odot{\prod\limits_{j=1}^q}^{\circ }\ell_{i_z}({\prod\limits_{i=1}^q}^{\ast}\phi(B))).\\
\end{split}
\end{equation*}

By Lemma \ref{l2}, we have
\begin{equation*}
\begin{split}
C
=& {\sum\limits_{i_1=1}^{h}}^{\varphi_0}{\sum\limits_{i_2=1}^{h}}^{\varphi_0}
\ldots{\sum\limits_{i_q=1}^{h}}^{\varphi_0}{\prod\limits_{z=1}^q}^{\circ }\varphi_{i_z}
(({\prod\limits_{j=q}^1}^\ast\ell_{i_j}\psi(A))
\odot({\prod\limits_{j=q}^1}^\ast\ell_{i_j}\phi(B)))\\
= & {\sum\limits_{i_1=1}^{h}}^{\varphi_0}{\sum\limits_{i_2=1}^{h}}^{\varphi_0}
\ldots{\sum\limits_{i_q=1}^{h}}^{\varphi_0}{\prod\limits_{z=1}^q}^{\circ }\varphi_{i_z}
(({\prod\limits_{j=q}^{1}}^{\ast}\psi_{i_j}(A))
\odot({\prod\limits_{j=q}^1}^{\ast}\phi_{i_j}(B)))\\
= & {\sum\limits_{i_1=1}^{h}}^{\varphi_0}{\sum\limits_{i_2=1}^{h}}^{\varphi_0}
\ldots{\sum\limits_{i_q=1}^{h}}^{\varphi_0}{\prod\limits_{z=1}^q}^{\circ }\varphi_{i_z}
(({\prod\limits_{j=q}^1}^{\ast}\psi_{i_j,2}\psi_{i_j,1}(A))
\odot({\prod\limits_{j=q}^1}^{\ast}\phi_{i_j,2}\phi_{i_j,1}(B))).\\
\end{split}
\end{equation*}

By Lemma \ref{l3}, we have
\begin{equation}\label{e1}
C = {\sum\limits_{i_1=1}^{h}}^{\varphi_0}{\sum\limits_{i_2=1}^{h}}^{\varphi_0}
\ldots{\sum\limits_{i_q=1}^{h}}^{\varphi_0}{\prod\limits_{z=1}^q}^{\circ }\varphi_{i_z}
({\prod\limits_{j=q}^1}^{\ast}\psi_{i_j,2}({\prod\limits_{j=q}^1}^{\ast}\psi_{i_j,1}(A))
\odot({\prod\limits_{j=q}^1}^{\ast}\phi_{i_j,2}({\prod\limits_{j=q}^1}^{\ast}\phi_{i_j,1}(B)))).
\end{equation}
Let $S=(s_1,s_2,...,s_u), T = (t_1,t_2,...,t_v)$ where $s_i,t_j$ are matrices for $1\leq i\leq u, 1\leq j\leq v$.
When $u=1,v=1$, $t=1$, Equation \ref{e1} is equivalent following formula:
\begin{equation}\label{e2}
{\sum\limits_{i_{t}=1}^{h}}^{\varphi_0}{\sum\limits_{i_{t+1}=1}^{h}}^{\varphi_0}
\ldots{\sum\limits_{i_q=1}^{h}}^{\varphi_0}{\prod\limits_{k=1}^{q}}^{\circ}\varphi_{i_k}
({\prod\limits_{j=q}^1}^{\ast}\psi_{i_j,2}(\bigoplus\limits_{z=1}^{u}{\prod\limits_{j=q}^{{t}}}^{\ast}\psi_{i_j,1}(s_z))
\odot{\prod\limits_{j=q}^1}^{\ast}\phi_{i_j,2}(\bigoplus\limits_{z=1}^{v}{\prod\limits_{j=q}^{{t}}}^{\ast}\phi_{i_j,1}(t_z)))
\end{equation}
where $i_1, i_2, \ldots, i_{t-1}$ are specified numbers with $1\leq i_j\leq h$ for $1\leq j\leq t-1$.
Then, we compute Formula \ref{e2}.

First, we compute the situation that $t-1=q$, i.e., compute the following formula:

\begin{equation}\label{e0}
{\prod\limits_{k=1}^{q}}^{\circ}\varphi_{i_k}
({\prod\limits_{j=q}^1}^{\ast}\psi_{i_j,2}(S)
\odot{\prod\limits_{j=q}^1}^{\ast}\phi_{i_j,2}(T)).
\end{equation}

Let $x$ be the number of $j$ with $i_j=1$ for $1\leq j\leq q$.
Since $\varphi_i = \psi_{i,2} = \phi_{i,2}= I$ for $2\leq i\leq t-p_1$, it equals to compute the following equation:

$$
{\prod\limits_{k=1}^{x}}^{\circ}\varphi_{1}
({\prod\limits_{j=1}^x}^{\ast}\psi_{1,2}(S)
\odot{\prod\limits_{j=1}^x}^{\ast}\phi_{1,2}(T)).
$$
Let $\psi_{1,2}(S) = (s'_1,s'_2,...,s'_{p_1}), \phi_{1,2}(T) = (t'_1,t'_2,...,t'_{p_1})$.
Then,
\begin{equation}\label{er}
\begin{split}
{\prod\limits_{k=1}^{x}}^{\circ}\varphi_{1}
({\prod\limits_{j=1}^x}^{\ast}\psi_{1,2}(S)
\odot{\prod\limits_{j=1}^x}^{\ast}\phi_{1,2}(T))
= &{\prod\limits_{k=1}^{x}}^{\circ}\varphi_{1}(\bigoplus\limits_{z=1}^{p_1}{\prod\limits_{j=1}^{x-1}}^{\ast}\psi_{1,2}(s'_z)
\odot{\prod\limits_{j=1}^{x-1}}^{\ast}\phi_{1,2}(t'_z))\\
=&\varphi_{1}(\bigoplus\limits_{z=1}^{p_1}{\prod\limits_{k=1}^{x-1}}^{\circ}\varphi_{1}({\prod\limits_{j=1}^{x-1}}^{\ast}\psi_{1,2}(s'_z)
\odot{\prod\limits_{j=1}^{x-1}}^{\ast}\phi_{1,2}(t'_z))).
\end{split}
\end{equation}

The Algorithm \ref{alg:W4} is designed to compute Formula \ref{e0} by Equation \ref{er}.

\begin{algorithm}\label{AC}
\SetAlgoNoLine
\KwIn{S,T}
\eIf{length of $S$ and $T$ equal one}{
    return $S\odot T$
}
{
    Let $(s'_1, s'_2, ..., s'_{p_1})$ be $\psi_{1,2}(S)$\;
    Let $(t'_1, t'_2, ..., t'_{p_1})$ be $\phi_{1,2}(T)$\;
    \For{$i=1$ to $p_1$}{
        Compute $s'_i,t'_i$\;
        $z_i = AC(s'_i,t'_i)$\;
        {Add $z_i$ to $re$ as $re = \varphi_1(z_1,z_2,...,z_{p_1})$*\;}
    }
}
\Return{re}
\caption{AC}
\label{alg:W4}
\end{algorithm}

\vspace{-0.55cm}
\begin{scriptsize}
*: Add $z_i$ to $re$ as $re = \varphi_1(z_1,z_2,...,z_{p_1})$ means that renew $re$ as $re = re + \varphi_1(z_1,z_2,...,z_{p_1})$ with setting $z_1 = z_2 = ...= z_{i-1} = z_{i+1} =...= z_{p_1} = 0$.
\end{scriptsize}

\vspace{0.4cm}
Now, we show how to compute Formula \ref{e2}.
First, for a specified $i_t$ and $1\leq z\leq u$, let
$$\psi_{i_t,1}(s_z) = (s_{i_t,(z-1)*p_{i_t,1}+1}, s_{i_t,(z-1)*p_{i_t,1}+2}, ..., s_{i_t,(z-1)*p_{i_t,1}+p_{i_t,1}})$$
and
$$\phi_{i_t,1}(t_z) = (t_{i_t,(z-1)*p'_{i_t,1}+1}, t_{i_t,(z-1)*p'_{i_t,1}+2}, ..., t_{i_t,(z-1)*p_{i_t,1}+p'_{i_t,1}}).$$
Then, for a specified $i_t$, we have
$$\bigoplus\limits_{z=1}^{u}{\prod\limits_{j=q}^{{t}}}^{\ast}\psi_{i_j,1}(s_z) = \bigoplus\limits_{z=1}^{up_{i_t,1}}{\prod\limits_{j=q}^{{t+1}}}^{\ast}\psi_{i_j,1}(s_{i_t,z})$$
and
$$\bigoplus\limits_{z=1}^{v}{\prod\limits_{j=q}^{{t}}}^{\ast}\phi_{i_j,1}(t_z) = \bigoplus\limits_{z=1}^{vp'_{i_t,1}}{\prod\limits_{j=q}^{{t+1}}}^{\ast}\phi_{i_j,1}(t_{i_t,z}).$$

Let $A_{i_t}$ be

\begin{equation*}
{\sum\limits_{i_{t+1}=1}^{h}}^{\varphi_0}
...{\sum\limits_{i_q=1}^{h}}^{\varphi_0}{\prod\limits_{k=1}^{q}}^{\circ}\varphi_{i_k}
({\prod\limits_{j=q}^1}^{\ast}\psi_{i_j,2}(\bigoplus\limits_{z=1}^{up_{i_t,1}}{\prod\limits_{j=q}^{{t+1}}}^{\ast}\psi_{i_j,1}(s_{i_t,z}))
\odot{\prod\limits_{j=q}^1}^{\ast}\phi_{i_j,2}(\bigoplus\limits_{z=1}^{vp'_{i_t,1}}{\prod\limits_{j=q}^{{t+1}}}^{\ast}\phi_{i_j,1}(t_{i_t,z}))).
\end{equation*}

Then, ${\sum\limits_{i_{t}=1}^{h}}^{\varphi_0}A_{i_t}$ equals Formula \ref{e2}.

\begin{def1}
Let $\xi_1, \xi_2$ be two maps with two inputs,
$S = (s_i)_{1\times u}, T = (t_i)_{1\times v}$,
$a_i = \bigoplus\limits_{j=1}^u\psi_{i,1}(s_j), b_i = \bigoplus\limits_{j=1}^v\phi_{i,1}(t_j)$ for $1\leq i\leq h$,
and $c_1 = \xi_1(a_1,b_1)$, $c_i = \xi_2(a_i, b_i)$ for $2\leq i\leq h$.
Define $\mathcal{F}(\xi_1, \xi_2, S, T) = {\sum\limits_{i=1}^h}^{\varphi_0}c_i$.

\end{def1}

Let $BC(\bigoplus\limits_{z=1}^{up_{i_t,1}}(s_{i_t,z}), \bigoplus\limits_{z=1}^{vp'_{i_t,1}}(t_{i_t,z})) := A_{i_t}$.
Notice that ${\sum\limits_{i_{t}=1}^{h}}^{\varphi_0}A_{i_t}$ equals $\mathcal{F}(BC,BC,S,T)$.

Let $\varphi_{1}(\psi_{1,2}\odot \phi_{1,2})(A',B')$ be the map $\varphi_1(\psi_{1,2}(A')\odot \phi_{1,2}(B'))$ where $A'\in R^{m_0\times n_0}, B'\in R^{n_0\times k_0}$, $\mathcal{A}$ be an algorithm computing $\mathcal{F}(\varphi_1(\psi_{1,2}\odot \phi_{1,2}), \odot, A', B')$.
Let $\mathcal{F}_{\mathcal{A}}(\xi_1, \xi_2, S, T)$ be an algorithm obtained by replacing $a'_i = \psi_{i,1}(A'), b'_i = \phi_{i,1}(B')$ with $a_i = \bigoplus\limits_{j=1}^u\psi_{i,1}(s_j), b_i = \bigoplus\limits_{j=1}^v\phi_{i,1}(t_j)$ $1\leq i\leq h$,
$c'_i = \psi_{i,1}(A')\odot \phi_{i,1}(B')$ with $c_i = \xi_2(a_i, b_i)$ for $2\leq i\leq h$,
$c'_1 = \varphi_1(\psi_{1,2}(a'_1)\odot \phi_{1,2}(b'_1))$ with $c_1 = \xi_1(a_1, b_1)$,
and $\sum\limits_{i=1}^hc'_i$ with $\sum\limits_{i=1}^h c_i$.

We present the Algorithm \ref{alg:W5} to compute Formula \ref{e2}.

We will present another algorithms(Algorithm \ref{CC},\ref{DC}), where Algorithm \ref{CC} is obtained from Algorithm \ref{alg:W5} by modifying the recursive exit and Algorithm \ref{DC} is similar with Algorithm \ref{alg:W4}.
In Algorithm \ref{CC}, we call Algorithm \ref{DC} when $BC(s_1,t_1)$ can be computed
within memory $M$.
In next section, we will give the needed memory size of $BC(S,T)$.

\break
\begin{multicols}{2}
\begin{minipage}{6cm}
\begin{algorithm}[H]
\SetAlgoNoLine
\KwIn{S,T}
\eIf{$s_1$ of $S$ is a number instead of a matrix}{
    $re = AC(S,T)$\;
}{
    $re = \mathcal{F}_{A}(BC, BC, S, T)$
}
\Return{re}
\caption{BC}
\label{alg:W5}
\end{algorithm}
\begin{algorithm}[H]\label{CC}
\SetAlgoNoLine
\KwIn{S,T}
\eIf{$BC(s_1, t_1)$ can be computed within memory $M$} {
    $re = DC(S, T)$\;
}{
    $re = \mathcal{F}_{A}(CC, CC, S, T)$
}
\Return{re}
\caption{CC}
\label{alg:W7}
\end{algorithm}
\end{minipage}
$ $

$ $

$ $
\begin{minipage}{6cm}
\begin{algorithm}[H]\label{DC}
\SetAlgoNoLine
\KwIn{S,T}
\eIf{length of $S$ and $T$ equal one}{
    \Return{$BC((S),(T))$}\;
}
{
    Let $(s'_1, s'_2, ..., s'_{p_1})$ be $\psi_{1,2}(S)$\;
    Let $(t'_1, t'_2, ..., t'_{p_1})$ be $\phi_{1,2}(T)$\;
    \For{$i=1$ to $p_1$}{
        Computing $s'_i,t'_i$\;
        $z_i = DC(s'_i,t'_i)$\;
        Add $z_i$ to $re$ as $re = \varphi_1(z_1,z_2,...,z_{p_1})$\;
    }
}
\Return{re}
\caption{DC}
\end{algorithm}
\end{minipage}
\end{multicols}
\subsection{Combine with alternative basis}
It is clear that our method can be combined with Karstadt-Schwartz's method\cite{karstadt2020matrix}.
Let $\gamma_1, \gamma_2, \gamma_3$ be the basis transformation functions.
We give the combined algorithms as following:

\begin{multicols}{2}
\begin{minipage}{6cm}
\begin{algorithm}[H]\label{fin1}
\SetAlgoNoLine
\KwIn{$A\in R^{n\times m}, B\in R^{m\times k}$}
{
$\tilde{A} = \gamma_1(A)$\;
$\tilde{B} = \gamma_2(B)$\;
$\tilde{C} = BC((\tilde{A}),(\tilde{B}))$\;
$C = \gamma_3^{-1}(\tilde{C})$\;
}
\Return{re}
\caption{BC*}
\label{alg:W6}
\end{algorithm}
\end{minipage}
$ $

$ $

$ $
\begin{minipage}{6cm}
\begin{algorithm}[H]\label{fin2}
\SetAlgoNoLine
\KwIn{$A\in R^{n\times m}, B\in R^{m\times k}$}
{
$\tilde{A} = \gamma_1(A)$\;
$\tilde{B} = \gamma_2(B)$\;
$\tilde{C} = CC((\tilde{A}),(\tilde{B}))$\;
$C = \gamma_3^{-1}(\tilde{C})$\;
}
\Return{re}
\caption{CC*}
\label{alg:W6}
\end{algorithm}
\end{minipage}
\end{multicols}
\section{complexity}
Recall that $\varphi_i = \psi_{i,2} = \phi_{i,2}= I ,p_i= q_i = p_{i,1} = p'_{i,1} = 1$ for $2\leq i\leq h$, $h = t - p_1+1$, and $A\in R^{n_0^q\times m_0^q}, B\in R^{m_0^q\times k_0^q}, C\in R^{n_0^q\times k_0^q}$.

For convenient, we assume that there are $\ell_{1,1}, \ell_{2,1}, \ell_{3,1}$ linear operations in $\psi_{1,2}$, $\phi_{1,2}$ and $\varphi_1$ respectively.
Let $\ell_{1,3} = p_{1,1}, \ell_{2,3} = p'_{1,1}, \ell_{3,3} = q_{1}, \ell_{1,4} = m_0n_0, \ell_{2,4} = n_0k_0, \ell_{3,4} = m_0k_0, u_0 = \max\{m_0n_0, n_0k_0, m_0k_0\}, r = \max\{p_{1,1}, p'_{1,1}, q_1\}$.

Let $\mathcal{A}$ be an algorithm to compute $\mathcal{F}(\psi_{1,2}\odot \phi_{1,2}, \odot, A', B')$ where $A'\in R^{m_0\times n_0}, B'\in R^{n_0\times k_0}$ and $C' = \mathcal{F}(\psi_{1,2}\odot \phi_{1,2}, \odot, A', B')$.
And let $\ell_{1,2}$ be total linear operations of $\mathcal{A}$ in computing $a'_i = \psi_{i,1}(A')$ for $1\leq i\leq h$, $\ell_{2,2}$ be total linear operations of $\mathcal{A}$ in computing $b'_i = \phi_{i,1}(B')$ for $1\leq i\leq h$ and $\ell_{3,2}$ be linear operations of $\mathcal{A}$ in computing ${\sum\limits_{i=1}^h}^{\varphi_{0}}c'_i$ where $c'_1 = \psi_{1,2}(a'_1)\odot \phi_{1,2}(b'_1), c'_{2} = a'_i\odot b'_i$ for $2\leq i\leq h$.
We assume that $\mathcal{A}$ needs additional $\alpha_1z_1 + \alpha_2z_2 + \alpha_3z_3$ size of memory to compute
$\mathcal{F}(\psi_{1,2}\odot \phi_{1,2}, \odot, A', B')$ where $z_1, z_2, z_3$ are the size of elements in $A', B'$ and $C'$ respectively.

\begin{observation}\label{ok1}
The arithmetic complexity of Algorithm \ref{fin1}(Algorithm \ref{fin2}) is the sum of arithmetic complexity in Algorithm \ref{alg:W5}(Algorithm \ref{CC}) and basis transformations.
The IO complexity of Algorithm \ref{fin1}(Algorithm \ref{fin2}) is the sum of IO complexity in Algorithm \ref{alg:W5}(Algorithm \ref{CC}) and basis transformations.

\end{observation}

\begin{lemma}\cite{karstadt2020matrix}\label{lemmak1}
Let $A\in R^{n\times n}$ and $\psi:R^{n_0\times n_0} \rightarrow R^{n_0\times n_0}$ be a basis transformation. The arithmetic complexity and IO complexity of $\psi$ on $A$ are $\frac{q}{n_0^2}n^2\log_{n_0}{n}$ and $\frac{3q}{n_0^2}n^2\log_{n_0}(\sqrt{2}\frac{n}{\sqrt{M}})+2M$ respectively where $q$ is the number of linear operations of $\psi$.
\end{lemma}

By Observation \ref{ok1} and Lemma \ref{lemmak1}, we will only study the arithmetic complexity and IO complexity of Algorithms \ref{alg:W5}, \ref{CC}.

For convenient, define a $(h,x)$-instance as $(S=(s_i)_{1\times \ell_{1,3}^x}, T = (t_i)_{1\times \ell_{2,3}^x})$ where $s_i\in R^{n_0^h\times m_0^h}$ for $1\leq i\leq \ell_{1,3}^x$, $t_i\in R^{m_0^h\times k_0^h}$ for $1\leq i\leq \ell_{2,3}^x$.
Specially, we call a $(0,x)$-instance as $x$-instance.
\subsection{Preliminary Of Complexity}

In this subsection, we discuss how much memory size it needs to compute Algorithm \ref{alg:W4} and Algorithm \ref{alg:W5}.
\begin{lemma}\label{IOL0}Algorithm \ref{alg:W4} can compute an $x$-instance within memory $M = \sum\limits_{i=1}^3\frac{\ell_{i,3}^{x+1}}{\ell_{i,3}-1}$.
\end{lemma}
\begin{proof}
We prove this by induction on $x$.
When $x=1$, since $\frac{\ell_{i,3}^2}{\ell_{i,3}-1}\geq \ell_{i,3}  + 1$ for $1\leq i \leq 3$, there are space storing $s'_i$ for $1\leq i\leq p_1$ where $(s'_1, s'_2, ..., s'_{p_1}) = \psi_{1,2}(S)$. Similarly, there is space storing $t'_i$ and $s'_i\odot t'_i$ for $1\leq i\leq p_1$. By Algorithm \ref{alg:W4}, it holds when $x=1$.
With assuming that it holds on $x-1$, we prove that it holds on $x$.
By Algorithm \ref{alg:W4}, we need to compute $s'_i,t'_i$, $z_i = AC(s'_i,t'_i)$ and add $z_i$ to $re$ as $re = \varphi_1(z_1,z_2,...,z_{p_1})$  for $1\leq i\leq p_1$.
Since length of $S$ is $\ell_{1,3}^x$, length of $T$ is $\ell_{2,3}^{x}$ and length of $AC(S,T)$ is $\ell_{3,3}^{x}$, there is $\sum\limits_{i=1}^3\frac{\ell_{i,3}^x}{\ell_{i,3}-1}$ memory size after storing $S, T, AC(S,T)$.
By the induction hypothesis, it is enough for computing $AC(s_i,t_i)$ for $1\leq i\leq p_1$.
So it holds on $x$.
This lemma holds.
\end{proof}

Let $\beta_i = \max\{\frac{\alpha_{i}}{\ell_{i,4}-\ell_{i,3}}, \frac{1}{\ell_{i,3}-1}\}, \beta=\sum\limits_{i=1}^3(1+\beta_i)$.
\begin{lemma}\label{IOL1}
Algorithm \ref{alg:W5} can compute a $(h,x)$-instance within memory $M = \sum\limits_{i=1}^3(\ell_{i,4}^h\ell_{i,3}^x+\beta_i\ell_{i,4}^h\ell_{i,3}^x)$.
\end{lemma}
\begin{proof}
We also prove this by induction on $h$.
When $h=0$, Algorithm \ref{alg:W5} calls Algorithm \ref{alg:W4}.
Since $\sum\limits_{i=1}^3(\ell_{i,3}^x+\beta_i\ell_{i,3}^x)\geq \sum\limits_{i=1}^3\ell_{i,3}^x\frac{\ell_{i,3}}{\ell_{i,3}-1}$, it holds by Lemma \ref{IOL0}.
With assuming that it holds on $h-1$, we prove that it holds on $h$.
Recall that $\mathcal{F}_{\mathcal{A}}(\xi_1,\xi_2,S,T)$ is obtained from $\mathcal{A}$ by replacing $a_i =  \psi_{i,1}(A'), b_i = \phi_{i,1}(B')$ with $a_i = \bigoplus\limits_{j=1}^u\psi_{i,1}(s_j), b_i = \bigoplus\limits_{j=1}^v\phi_{i,1}(t_j)$ for $1\leq i\leq t-p_1+1$, $c_1 = \varphi_1(\psi_{1,2}\odot \phi_{1,2})(a_1,b_1)$ with $c_1 = \xi_1(a_1,b_1)$, and $c_i = a_i\odot b_i$ with $c_i = \xi_2(a_i,b_i)$ for $2\leq i\leq t-p_1+1$.
Since $\mathcal{A}$ needs additional $\alpha_1z_1 + \alpha_2z_2 + \alpha_3z_3$ memory size to compute
$\mathcal{F}(\psi_{1,2}\odot \phi_{1,2}, \odot, A', B')$ where $z_1, z_2, z_3$ are the size of elements in $A', B'$ and $C'$ respectively,
$\mathcal{F}_{\mathcal{A}}(\varphi_1(\psi_{1,2}\odot \phi_{1,2}),\odot,S,T)$ needs additional $\sum\limits_{i=1}^3\alpha_i\ell_{i,4}^{h-1}\ell_{i,3}^{x}$ memory size.
Since $S,T$ and $\mathcal{F}_{\mathcal{A}}(\varphi_1(\psi_{1,2}\odot \phi_{1,2}),\odot,S,T)$ can be stored within $\sum\limits_{i=1}^3\ell_{i,4}^h\ell_{i,3}^x$ memory size, $\mathcal{F}_{\mathcal{A}}(\varphi_1(\psi_{1,2}\odot \phi_{1,2}),\odot,S,T)$ can be computed within $\sum\limits_{i=1}^3(\ell_{i,4}^h\ell_{i,3}^x+\alpha_i\ell_{i,4}^{h-1}\ell_{i,3}^{x})$ memory size.

Comparing $\mathcal{F}_{\mathcal{A}}(\varphi_1(\psi_{1,2}\odot \phi_{1,2}),\odot,S,T)$ and $\mathcal{F}_{\mathcal{A}}(BC,BC,S,T)$, $\mathcal{F}_{\mathcal{A}}(BC,BC,S,T)$ replaces $\varphi_1(\psi_{1,2}\odot \phi_{1,2}),\odot$ operations with $BC$ operation.
Notice that $\varphi_1(\psi_{1,2}\odot \phi_{1,2})(a_1,b_1),a_1\odot b_1$ need at least $\sum\limits_{i=1}^3\ell_{i,4}^{h-1}\ell_{i,3}^{x+1}$ memory size
and $\varphi_1(\psi_{1,2}\odot \phi_{1,2})(a_j,b_j),a_j\odot b_j$ need at least $\sum\limits_{i=1}^3\ell_{i,4}^{h-1}\ell_{i,3}^{x}$ memory size for $2\leq j\leq t-p_1+1$.
Since $BC(a_1,b_1)$ can be computed within $\sum\limits_{i=1}^3(\ell_{i,4}^{h-1}\ell_{i,3}^{x+1}+\beta_i\ell_{i,4}^{h-1}\ell_{i,3}^{x+1})$ memory size and $BC(a_j,b_j)$ can be computed within $\sum\limits_{i=1}^3(\ell_{i,4}^{h-1}\ell_{i,3}^{x}+\beta_i\ell_{i,4}^{h-1}\ell_{i,3}^{x})$ memory size for $2\leq j\leq t-p_1+1$ by induction hypothesis,
$\mathcal{F}_{\mathcal{A}}(BC,BC,S,T)$ can be computed within $\sum\limits_{i=1}^3(\ell_{i,4}^h\ell_{i,3}^x+\alpha_i\ell_{i,4}^{h-1}\ell_{i,3}^{x})
-\sum\limits_{i=1}^3\ell_{i,4}^{h-1}\ell_{i,3}^{x+1}+\sum\limits_{i=1}^3(\ell_{i,4}^{h-1}\ell_{i,3}^{x}+\beta_i\ell_{i,4}^{h-1}\ell_{i,3}^{x})$
memory size.
Since $\sum\limits_{i=1}^3(\ell_{i,4}^h\ell_{i,3}^x+\alpha_i\ell_{i,4}^{h-1}\ell_{i,3}^{x})
-\sum\limits_{i=1}^3\ell_{i,4}^{h-1}\ell_{i,3}^{x+1}+\sum\limits_{i=1}^3(\ell_{i,4}^{h-1}\ell_{i,3}^{x+1}+\beta_i\ell_{i,4}^{h-1}\ell_{i,3}^{x+1})\leq \sum\limits_{i=1}^3(\ell_{i,4}^h\ell_{i,3}^x+\beta_i\ell_{i,4}^h\ell_{i,3}^x)$, this lemma holds.

\end{proof}

Let
$h' = \frac{\ln{M}-\ln{\beta}}{\ln{u_0}}$.
Notice that Algorithm \ref{alg:W5} can be compute a $(h,0)$-instance within memory $M$ when $h\leq h'$.

\subsection{Arithmetic Complexity}

\begin{lemma}\label{a1}
The arithmetic complexity of computing a $x$-instance by Algorithm \ref{alg:W4} is $\mathcal{T}(x) = p_1^x +  \sum\limits_{i=1}^3\ell_{i,1}\sum\limits_{j=0}^{x-1}\ell_{i,3}^{x-1-j}p_1^{j}$.
\end{lemma}
\begin{proof}
We prove this by induction on $x$.
When $x=0$, we have $\mathcal{T}(x) = 1$. This lemma holds.
With assuming it holds on $x-1$, we prove that it holds on $x$.
By Algorithm \ref{alg:W4}, we have $\mathcal{T}(x) = p_1\mathcal{T}(x-1) + \sum\limits_{i=1}^3\ell_{i,1}\ell_{i,3}^{x-1}$.
It deduces that $\mathcal{T}(x) = p_1^x + \sum\limits_{i=1}^3\ell_{i,1}\sum\limits_{j=0}^{x-1}\ell_{i,3}^{x-i-1}p_1^{i}$.
This lemma holds.
\end{proof}
Now, denote following equations:

\begin{equation*}
\begin{split}
P_{1,h,x} &= t^{h}\mathcal{T}(x)\\
P_{2,h,x} &= \sum\limits_{i=1}^3\ell_{i,1}\ell_{i,3}^x\sum\limits_{j=0}^{h-1}t^j{(t-p_1+\ell_{i,3})}^{h-1-j}\\
P_{3,h,x} &= \sum\limits_{i=1}^{3}\ell_{i,2}\ell_{i,3}^x\sum\limits_{j=0}^{h-1}\ell_{i,4}^j{(t-p_1+\ell_{i,3})}^{h-1-j}
\end{split}
\end{equation*}
\begin{thm}\label{a2}

The arithmetic complexity of computing a $(h,x)$-instance by Algorithm \ref{alg:W5} is
$\mathcal{T}(h,x) = P_{1,h,x} + P_{2,h,x} + P_{3,h,x}$.
\end{thm}
\begin{proof}
We prove this by induction on $h$. Since Algorithm \ref{alg:W5} calls Algorithm \ref{alg:W4} when $h = 0$,
this lemma holds when $h = 0$.

With assuming it holds on $h-1$, we prove that it holds on $h$.
Notice that $\mathcal{F}_{\mathcal{A}}(BC,BC,S,T)$ will compute $a_i = \bigoplus\limits_{j=1}^{\ell_{1,3}^x}\psi_{i,1}(s_j), b_i = \bigoplus\limits_{j=1}^{\ell_{2,3}^x}\phi_{i,1}(t_j)$, $c_i = BC(a_i,b_i)$ for $1\leq i\leq t - p_1 + 1$ and ${\sum\limits_{i=1}^{t-p_1+1}}^{\varphi_0}c_i$.
Since the length of $a_1,b_1$ are $\ell_{1,3}^{x+1}, \ell_{2,3}^{x+1}$ respectively and the length of $a_i,b_i$ are $\ell_{1,3}^{x}, \ell_{2,3}^{x}$ respectively for $2\leq i\leq t-p_1+1$, the arithmetic complexity of computing $BC(a_i,b_i)$ for $1\leq i\leq t-p_1+1$ is $(t-p_1)\mathcal{T}(h-1,x) + \mathcal{T}(h-1,x+1)$.
And the arithmetic complexity of computing $a_i,b_i, \sum\limits_{i=1}^{t-p_1+1}c_i$ are $\ell_{1,2}\ell_{1,3}^x\ell_{1,4}^{h-1}, \ell_{2,2}\ell_{2,3}^x\ell_{2,4}^{h-1}, \ell_{3,2}\ell_{3,3}^x\ell_{3,4}^{h-1}$ respectively. Then,
$$\mathcal{T}(h,x) = (t-p_1)\mathcal{T}(h-1,x) + \mathcal{T}(h-1,x+1) +
\sum\limits_{i=1}^3\ell_{i,2}\ell_{i,3}^x\ell_{i,4}^{h-1}$$
First, we have
\begin{equation*}
\begin{split}
&(t-p_1)P_{1,h-1,x} + P_{1,h-1,x+1} \\
=& (t-p_1)t^{h-1}\mathcal{T}(x) + t^{h-1}\mathcal{T}(x+1)\\
=& (t-p_1)t^{h-1}\mathcal{T}(x) + p_1t^{h-1}\mathcal{T}(x) + t^{h-1}\sum\limits_{i=1}^3\ell_{i,1}\ell_{i,3}^{x}\\
=& t^h\mathcal{T}(x) + t^{h-1}\sum\limits_{i=1}^3\ell_{i,1}\ell_{i,3}^{x}\\
=& P_{1,h,x} +  t^{h-1}\sum\limits_{i=1}^3\ell_{i,1}\ell_{i,3}^{x}\\
\end{split}
\end{equation*}
Second, we have
\begin{equation*}
\begin{split}
&(t-p_1)P_{2,h-1,x} + P_{2,h-1,x+1} + t^{h-1}\sum\limits_{i=1}^3\ell_{i,1}\ell_{i,3}^{x}\\
=& \sum\limits_{i=1}^3\ell_{i,1}\ell_{i,3}^x\sum\limits_{j=0}^{h-2}t^j{(t-p_1+\ell_{i,3})}^{h-1-j} + t^{h-1} \sum\limits_{i=1}^3\ell_{i,1}\ell_{i,3}^{x}\\
=& P_{2,h,x}
\end{split}
\end{equation*}

Last, we have

\begin{equation*}
\begin{split}
&(t-p_1)P_{3,h-1,x} + P_{3,h-1,x+1} + \sum\limits_{i=1}^3\ell_{i,2}\ell_{i,3}^x\ell_{i,4}^{h-1}\\
=&\sum\limits_{i=1}^3\ell_{i,2}\ell_{i,3}^x\sum\limits_{j=0}^{h-2}\ell_{i,4}^{j}{(t-p_1+\ell_{i,3})}^{h-2-j}
+\sum\limits_{i=1}^3\ell_{i,2}\ell_{i,3}^x\ell_{i,4}^{h-1}\\
=& P_{3,h,x}
\end{split}
\end{equation*}

Then, we have
\begin{equation*}
\begin{split}
\mathcal{T}(h,x) &= (t-p_1)\mathcal{T}(h-1,x) + \mathcal{T}(h-1,x+1) +
\sum\limits_{i=1}^3\ell_{i,2}\ell_{i,3}^x\ell_{i,4}^{h-1}\\
& = (t-p_1)\sum\limits_{i=1}^3P_{i,h-1,x} + \sum\limits_{i=1}^3P_{i,h-1,x+1} +
\sum\limits_{i=1}^3\ell_{i,2}\ell_{i,3}^x\ell_{i,4}^{h-1}\\
& = P_{1,h,x} + t^{h-1}\sum\limits_{i=1}^3\ell_{i,1}\ell_{i,3}^{x} + (t-p_1)\sum\limits_{i=2}^3P_{i,h-1,x} + \sum\limits_{i=2}^3P_{i,h-1,x+1} + \sum\limits_{i=1}^3\ell_{i,2}\ell_{i,3}^x\ell_{i,4}^{h-1}\\
& = P_{1,h,x} + P_{2,h,x} + (t-p_1)P_{3,h-1,x} + P_{3,h-1,x+1} + \sum\limits_{i=1}^3\ell_{i,2}\ell_{i,3}^x\ell_{i,4}^{h-1}\\
& = P_{1,h,x} + P_{2,h,x} + P_{3,h,x}
\end{split}
\end{equation*}
Thus, this lemma holds.
\end{proof}

\begin{cor}\label{c1}
The leading coefficient of arithmetic complexity for Algorithm \ref{alg:W5} with a $(h,0)$-instance is $\sum\limits_{i=1}^3\frac{\ell_{i,1}}{p_1-\ell_{i,3}} + 1$ only depended on $\varphi_1, \phi_{1,1}, \psi_{1,1}$.
\end{cor}

\begin{lemma}

The arithmetic complexity of computing a $(h,x)$-instance by Algorithm \ref{DC} is $\mathcal{B}(h,x) = \sum\limits_{i=1}^3\ell_{i,1}\ell_{i,4}^h\sum\limits_{j=0}^{x-1}\ell_{i,3}^{x-1-j}p_1^{j}
+ p_1^x\mathcal{T}(h,0)$.
\end{lemma}
\begin{proof}
We prove this by induction on $x$.
Since Algorithm \ref{DC} calls Algorithm \ref{alg:W5} when $x=0$, this lemma holds when $x = 0 $.
With assuming it holds on $x-1$, we prove that it holds on $x$.
By Algorithm \ref{DC}, we have $\mathcal{B}(h,x) = p_1\mathcal{B}(h,x-1) + \sum\limits_{i=1}^3\ell_{i,1}\ell_{i,4}^h\ell_{i,3}^{x-1}$.
It deduces that $\mathcal{B}(h,x) = \sum\limits_{i=1}^3\ell_{i,1}\ell_{i,4}^h\sum\limits_{j=0}^{x-1}\ell_{i,3}^{x-j-1}p_1^{j} + p_1^{x}\mathcal{T}(h,0)$.
This lemma holds.
\end{proof}

Let $P'_{1,h+h',x} = t^{h}\mathcal{B}(h',x), P'_{2,h+h',x} = \ell_{i,4}^{h'}P_{2,h,x}. P'_{3,h+h',x} = \ell_{i,4}^{h'}P_{3,h,x}$.

\begin{thm}\label{a4}

The arithmetic complexity of computing a $(h,x)$-instance by Algorithm \ref{CC}
is $\mathcal{C}(h,x) = P'_{1,h,x} + P'_{2,h,x} + P'_{3,h,x}$.
\end{thm}
\begin{proof}
The proof is similar with the proof of Theorem \ref{a2}.
\end{proof}
\begin{cor}\label{c1}
The leading coefficient of arithmetic complexity for Algorithm \ref{alg:W7} with input $A\in R^{m_0^h\times n_0^h}, B\in R^{n_0^h\times k_0^h}$ is $1 + \frac{P_{3,h',0}}{t^{h'}} +
\sum\limits_{i=1}^3(1-{(\frac{t-p_1+\ell_{i,3}}{t})}^{h'} + {(\frac{\ell_{i,4}}{t})}^{h'})\frac{\ell_{i,1}}{p_1-\ell_{i,3}}
$.

\end{cor}
\subsection{IO Complexity}

Let $x'= \frac{\ln{M(r-1)}-\ln{3r}}{\ln{r}}$.
Notice that Algorithm \ref{alg:W4} can compute an $x$-instance within $M$ memory when $x\leq x'$ by Lemma \ref{IOL0}.
\begin{lemma}\label{oIOAC}
The IO complexity of computing a $x$-instance with $x\geq x'$ by Algorithm \ref{alg:W4}
is $M(x) = p_1^{x-x'}M + 3\sum\limits_{i=1}^3\ell_{i,1}\ell_{i,3}^{x'}\frac{p_1^{x-x'}-\ell_{i,3}^{x-x'}}{p_1-\ell_{i,3}}$.
\end{lemma}
\begin{proof}
For an $x$-instance with $x<x'$, it can be computed within memory $M$ by Lemma \ref{IOL0}.
Then, when $x\leq x'$, the IO complexity is $\sum\limits_{i=1}^3\ell_{i,3}^x$ less than $M$.
When $x>x'$, $M(x) = p_1M(x-1) + 3\sum\limits_{i=1}^3\ell_{i,1}\ell_{i,3}^{x-1}$.
Thus $M(x) = p_1^{x-x'}M + 3\sum\limits_{i=1}^3\ell_{i,1}\sum\limits_{j=x'}^{x-1}p_1^{x-1-j}\ell_{i,3}^{j}= p_1^{x-x'}M + 3\sum\limits_{i=1}^3\ell_{i,1}\ell_{i,3}^{x'}\frac{p_1^{x-x'}-\ell_{i,3}^{x-x'}}{p_1-\ell_{i,3}}$.
\end{proof}

For convenient, we denote
\begin{equation*}
M(x) =  \left\{
\begin{aligned}
&p_1^{x-x'}M + 3\sum\limits_{i=1}^3\ell_{i,1}\ell_{i,3}^{x'}\frac{p_1^{x-x'}-\ell_{i,3}^{x-x'}}{p_1-\ell_{i,3}}, & x\geq x' \\
&p_1^{x-x'}M, & x < x'
\end{aligned}
\right.
\end{equation*}

Note that $M(x)$ is not the IO complexity of Algorithm \ref{alg:W4} when $x<x'$.

\begin{lemma}
The IO complexity of computing a $(h',x)$-instance by Algorithm \ref{DC}
is $M'(x) = p_1^x\sum\limits_{i=1}^3\ell_{i,4}^{h'} + 3\sum\limits_{i=1}^3\ell_{i,1}\ell_{i,4}^{h'}\sum\limits_{j=0}^{x-1}\ell_{i,3}^{j}p_1^{x-1-j}$.
\end{lemma}
\begin{proof}
By Lemma \ref{IOL1}, $M'(0) = \sum\limits_{i=1}^3\ell_{i,4}^{h'}\leq M$, since Algorithm \ref{DC} calls Algorithm \ref{alg:W5} when $x = 0$.
And $M'(x) = p_1M'(x-1) + 3\sum\limits_{i=1}^3\ell_{i,1}\ell_{i,3}^{x-1}\ell_{i,4}^{h'}$ when $x>0$.
Then $M'(x) = p_1^x\sum\limits_{i=1}^3\ell_{i,4}^{h'} + 3\sum\limits_{i=1}^3\ell_{i,1}\ell_{i,4}^{h'}\sum\limits_{j=0}^{x-1}\ell_{i,2}^{j}p_1^{x-1-j}$.
\end{proof}

\begin{thm}\label{IOA3}

The IO complexity of computing a $(h,x)$-instance by Algorithm \ref{CC} is $M'(h,x) = t^{h-h'}M'(x) +3
\sum\limits_{i=1}^3\ell_{i,3}^x\ell_{i,4}^{h'}\sum\limits_{j=0}^{h-1-h'}(t-p_1+\ell_{i,3})^{h-h'-1-j}(\ell_{i,1}t^{j} + \ell_{i,2}\ell_{i,4}^{j})$.
\end{thm}
\begin{proof}
When $h\leq h'$, Algorithm \ref{CC} calls Algorithm \ref{DC}, then the IO complexity is $M'(x)$.
When $h>h'$, $M'(h,x) = (t-p_1)M'(h-1,x) + M'(h-1,x+1) + 3\sum\limits_{i=1}^{3}\ell_{i,2}\ell_{i,3}^{x}\ell_{i,4}^{h-1}$.
Similarly with the proof of Theorem \ref{a2}, $M'(h,x) = t^{h-h'}M'(x) +3
\sum\limits_{i=1}^3\ell_{i,3}^x\ell_{i,4}^{h'}\sum\limits_{j=0}^{h-1-h'}(t-p_1+\ell_{i,3})^{h-h'-1-j}(\ell_{i,1}t^{j} + \ell_{i,2}\ell_{i,4}^{j})$.
\end{proof}

Let
\begin{equation*}
\begin{split}
Z(h,x) &= \sum\limits_{i=1}^3(\ell_{i,4}^h\ell_{i,3}^x + \beta_i\ell_{i,3}^x\ell_{i,4}^h),\\
M_1(h,x) &= 3\sum\limits_{i=1}^3\ell_{i,1}\ell_{i,3}^x\sum\limits_{j=0}^{h-1}(t-p_1+\ell_{i,3})^{h-1-j}t^j,\\
M_2(h,x) &= 3 \sum\limits_{i=1}^3\ell_{i,1}\ell_{i,3}^x(\frac{\ell_{i,3}}{p_1})^{x'-x}\sum\limits_{j=0}^{h-1}(t-p_1+\ell_{i,3})^{h-1-j}t^j,\\
T_x(h,j)&\leq \left\{
\begin{aligned}
&3\sum\limits_{i=1}^3\ell_{i,2}{h\choose {x-j}}\frac{(t-p_1)^{h-x+j}}{t-p_1+\ell_{i,3}-\ell_{i,4}}\ell_{i,3}^{x}, & Z(0,x)\geq M \\
&3\sum\limits_{i=1}^3\ell_{i,2}{h\choose {x-j}}(\frac{\ell_{i,4}}{t-p_1})^{h'-\frac{\ln{r}}{\ln{u_0}}x}\frac{(t-p_1)^{h-x+j}}{t-p_1+\ell_{i,3}-\ell_{i,4}}\ell_{i,3}^{x}, & Z(0,x)<M\\
\end{aligned}
\right.,\\
M_3(h,x) &= \sum\limits_{i=x}^{h+x}T_{i}(h,x)-
3\sum\limits_{i=1}^3\ell_{i,3}^x\ell_{i,2}\frac{\ell_{i,4}^h}{t-p_1+\ell_{i,3}-\ell_{i,4}},\\
R(h,x)
&= \left\{
\begin{aligned}
&t^hM(x) + M_1(h,x) + M_3(h,x), & x\geq x' \\
&t^hM(x) + M_2(h,x) + M_3(h,x), & x < x'
\end{aligned}
\right..\\
\end{split}
\end{equation*}

\begin{observation}\label{oIO1}
$R(h,x)\geq (t-p_1)R(h-1,x) + R(h-1,x+1) + 3\sum\limits_{i=1}^3\ell_{i,2}\ell_{i,3}^x\ell_{i,4}^{h-1}$.
\end{observation}
\begin{proof}
It is clear to check that $T_i(h-1,x+1) + (t-p_1)T_i(h-1,x) = T_i(h,x)$ for $1+x\leq i\leq h+x-1$.
Then, $(t-p_1)\sum\limits_{i=x}^{h+x-1}T_{i}(h-1,x) + \sum\limits_{i=x+1}^{h+x}T_{i}(h-1,x+1) = \sum\limits_{i=x}^{h+x}T_i(h,x)$.
Since $(t-p_1)\sum\limits_{i=1}^3\ell_{i,3}^x\ell_{i,2}\frac{\ell_{i,4}^{h-1}}{t-p_1+\ell_{i,3}-\ell_{i,4}}
+\sum\limits_{i=1}^3\ell_{i,3}^{x+1}\ell_{i,2}\frac{\ell_{i,4}^{h-1}}{t-p_1+\ell_{i,3}-\ell_{i,4}}
-$ $\sum\limits_{i=1}^3\ell_{i,2}\ell_{i,3}^x\ell_{i,4}^{h-1}$
 $=\sum\limits_{i=1}^3\ell_{i,3}^x\ell_{i,2}\frac{\ell_{i,4}^h}{t-p_1+\ell_{i,3}-\ell_{i,4}}$,
$(t-p_1)M_3(h-1,x) + M_3(h-1,x+1) + 3\sum\limits_{i=1}^3\ell_{i,2}\ell_{i,3}^x\ell_{i,4}^{h-1}= M_3(h,x)$.

\textbf{Case 1}: $x\geq x'$.\\
First, we have $(t-p_1)t^{h-1}M(x) + t^{h-1}M(x+1) = t^{h}M(x) + 3t^{h-1}\sum\limits_{i=1}^3\ell_{i,1}\ell_{i,3}^{x-1}$.
Second, we have $(t-p_1)M_1(h-1,x) + M_1(h-1,x+1) + 3t^{h-1}\sum\limits_{i=1}^3\ell_{i,1}\ell_{i,3}^{x-1} = M_1(h,x)$.
Then, we have $R(h,x) =  (t-p_1)R(h-1,x) + R(h-1,x+1) + 3\sum\limits_{i=1}^3\ell_{i,2}\ell_{i,3}^x\ell_{i,4}^{h-1}$.

\textbf{Case 2}: $x<x'$.\\
First, $p_1M(x)= M(x+1)$ for $x\leq x'-1$.
Then, $(t-p_1)t^{h-1}M(x) + t^{h-1}M(x+1)= (t-p_1)t^{h-1}M(x) + p_1t^{h-1}M(x)=t^hM(x)$ for $x\leq x'-1$. Second,

\begin{equation*}
\begin{split}
&(t-p_1)M_2(h-1,x) + M_2(h-1,x+1)\\
\leq& 3\sum\limits_{i=1}^3\ell_{i,1}\ell_{i,3}^x(\frac{\ell_{i,3}}{p_1})^{x'-x-1}((t-p_1)\frac{\ell_{i,3}}{p_1} +\ell_{i,3})\sum\limits_{j=0}^{h-2}(t-p_1+\ell_{i,3})^{h-2-j}t^j\\
\leq & 3\sum\limits_{i=1}^3\ell_{i,1}\ell_{i,3}^x(\frac{\ell_{i,3}}{p_1})^{x'-x-1}((t-p_1)\frac{\ell_{i,3}}{p_1}+\ell_{i,3}\frac{\ell_{i,3}}{p_1}+\ell_{i,3}\frac{p_1 - \ell_{i,3}}{p_1})\sum\limits_{j=0}^{h-2}(t-p_1+\ell_{i,3})^{h-2-j}t^j\\
\leq & 3\sum\limits_{i=1}^3\ell_{i,1}\ell_{i,3}^x(\frac{\ell_{i,3}}{p_1})^{x'-x-1}(\frac{\ell_{i,3}}{p_1}\sum\limits_{j=0}^{h-2}(t-p_1+\ell_{i,3})^{h-1-j}t^j + \ell_{i,3}\frac{p_1 - \ell_{i,3}}{p_1}\sum\limits_{j=0}^{h-2}(t-p_1+\ell_{i,3})^{h-2-j}t^j)\\
\leq &
3\sum\limits_{i=1}^3(\ell_{i,1}\ell_{i,3}^x(\frac{\ell_{i,3}}{p_1})^{x'-x-1}(\frac{\ell_{i,3}}{p_1}\sum\limits_{j=0}^{h-2}(t-p_1+\ell_{i,3})^{h-1-j}t^j + \frac{\ell_{i,3}}{p_1}t^{h-1})\\
=& M_2(h,x)\\
\end{split}
\end{equation*}

Then, $(t-p_1)R(h-1,x)+R(h-1,x+1)+3\sum\limits_{i=1}^3\ell_{i,2}\ell_{i,3}^x\ell_{i,4}^{h-1}\leq R(h,x)$ for $x\leq x'-1$.
It is clear to check that $t^{h-1}M(x+1) + M_1(h-1,x+1) + (t-p_1)(t^{h-1}M(x)+M_2(h-1,x))\leq t^hM(x)+M_2(h,x)$ for $x'>x>x'-1$.
Then, $R(h,x)\geq (t-p_1)R(h-1,x)+R(h-1,x+1)+3\sum\limits_{i=1}^3\ell_{i,2}\ell_{i,3}^x\ell_{i,4}^{h-1}$ for $x<x'$.
So this observation holds.
\end{proof}

\begin{thm}\label{IOL2}
The IO complexity of computing a $(h,x)$-instance which can not computed within memory $M$ by Algorithm \ref{alg:W5}
is $M(h,x)\leq R(h,x)$.
\end{thm}
\begin{proof}
We prove this by induction on $h$.
First, we show that it holds in the base case for each $x$.

$\bullet$ For $x$ with $Z(0,x)\geq M$, the IO complexity is the IO complexity of Algorithm \ref{alg:W4} when $h = 0$, since Algorithm \ref{alg:W5} calls Algorithm \ref{alg:W4} when $h = 0$.
If $x\geq x'$, the base case is the $(0,x)$-instance.
By Lemma \ref{oIOAC}, IO complexity of Algorithm \ref{alg:W4} is $M(x)$.
It holds, since $R(0,x) \geq M(x)$.
If $x<x'$, the $(0,x)$-instance can be computed within memory $M$ by Lemma \ref{IOL0}.
Then, the base case is the $(1,x)$-instance.
Let $S = (s_i)_{1\times \ell_{1,3}^x}, T = (t_i)_{1\times \ell_{2,3}^x}$ be this $(1,x)$-instance and $a_i = \bigoplus\limits_{j=1}^{\ell_{1,3}^x}\psi_{i,1}(s_j),b_i = \bigoplus\limits_{j=1}^{\ell_{2,3}^x}\phi_{i,1}(s_j)$ $1\leq i\leq t-p_1+1$.
Then, IO complexity of $BC(a_i,b_i)$ is $0$ for $2\leq i\leq t-p_1+1$, since $(a_i,b_i)$ is $(0,x)$-instance for $2\leq i\leq t-p_1+1$ and it can be within memory $M$.
If $x+1 < x'$, $BC(a_1,b_1)$ can be computed within memory $M$ by Lemma \ref{IOL0}.
Then, either the IO complexity of $BC(a_1,b_1)$ is less than $R(0,x+1)$(when $x+1\geq x'$), or it is $0$(when $x+1<x'$).
Thus, $M(1,x)\leq (t-p_1)R(0,x) + R(0,x+1) + 3\sum\limits_{i=1}^3\ell_{i,2}\ell_{i,3}^x\ell_{i,4}^{h-1}$, since $R(0,x)\geq 0, R(0,x+1)\geq 0$ for $x$ with $Z(0,x)\geq M$.
What's more, $M(1,x)\leq R(1,x)$ by observation \ref{oIO1}.

$\bullet$ For $x$ with $Z(0,x)< M$, let $h$ satisfy that $Z(h,x) = M$.
The base case is the $(h+1,x)$-instance.
Let $S = (s_i)_{1\times \ell_{1,3}^x}, T = (t_i)_{1\times \ell_{2,3}^x}$ be this $(h+1,x)$-instance and $a_i = \bigoplus\limits_{j=1}^{\ell_{1,3}^x}\psi_{i,1}(s_j),b_i = \bigoplus\limits_{j=1}^{\ell_{2,3}^x}\phi_{i,1}(s_j)$ $1\leq i\leq t-p_1+1$.
Since $Z(h,x) = M$, $BC(a_i,b_i)$ can be computed within memory $M$ for $2\leq i\leq t-p_1+1$ by Lemma \ref{IOL1}.
Then, IO complexity of $BC(a_i,b_i)$ is $0$ for $2\leq i\leq t-p_1+1$.
Since $Z(h,x+1)>M$, the IO complexity of $BC(a_1,b_1)$ is less than $R(h,x+1)$ by the induction hypothesis.
Then, $M(h+1,x)\leq R(h,x+1) + 3\sum\limits_{i=1}^3\ell_{i,2}\ell_{i,3}^x\ell_{i,4}^{h}$.
Since $Z(h,x) = M$, $(3+\sum\limits_{i=1}^3\beta_i)u_0^{h}r^x\geq Z(h,x)=  M$, e.q, $h\geq h'-\frac{\ln{r}}{\ln{u_0}}x$.
Since $t-p_1> \ell_{i,4}$ for $1\leq i\leq 3$, $T_x(h,x)\geq 3\sum\limits_{i=1}^3\ell_{i,3}^x\ell_{i,2}\frac{\ell_{i,4}^h}{t-p_1+\ell_{i,3}-\ell_{i,4}}$. Then, $M_3(h,x)\geq 0$.
Thus, $R(h,x)\geq 0$.
Therefore, $M(h+1,x)\leq (t-p_1)R(h,x) + R(h,x+1) + 3\sum\limits_{i=1}^3\ell_{i,2}\ell_{i,3}^x\ell_{i,4}^{h}$.
By observation \ref{oIO1}, $M(h+1,x)\leq R(h+1,x)$.

Then it holds in the base case.
In other cases, by Algorithm \ref{alg:W5}, $M(h,x)\leq (t-p_1)M(h-1,x) + M(h-1,x+1)+3\sum\limits_{i=1}^3\ell_{i,2}\ell_{i,3}^x\ell_{i,4}^{h-1}$.
Since $(h,x)$-instance is not the base case, $Z(h-1,x)\geq M$. Then, $M(h-1,x)\leq R(h-1,x), M(h-1,x+1)\leq R(h-1,x+1)$ by the induction hypothesis.
By Observation \ref{oIO1}, $M(h,x)\leq R(h,x)$.

\end{proof}

\begin{cor}
Let $m_0=n_0=k_0=t_0$. The IO complexity of Algorithm \ref{alg:W5} with input $A,B\in R^{n\times n}$ is no more than $O(n^{{\log_{t_0}{t}}}M^{1-\log_{r}{p_1}})$.
\end{cor}
\begin{proof}
By Theorem \ref{IOL2}, we have $M(h,0)\leq t^hM(0) + M_2(h,0) + M_3(h,0)$.
Since $M_3(h,0) = \sum\limits_{i=0}^{h}T_{i}(h,0)-
3\sum\limits_{i=1}^3\ell_{i,3}^x\ell_{i,2}\frac{\ell_{i,4}^h}{t-p_1+\ell_{i,3}-\ell_{i,4}}$ and $T_{x}(h,0)\leq 3\sum\limits_{i=1}^3\ell_{i,2}{h\choose {x}}\frac{(t-p_1)^{h-x}}{t-p_1+\ell_{i,3}-\ell_{i,4}}\ell_{i,3}^{x}$ for $0\leq x\leq h$.
We have $M_3(h,0)\leq 3\sum\limits_{i=1}^3\frac{(t-p_1+\ell_{i,3})^{h} - \ell_{i,4}^h}{t-p_1+\ell_{i,3}-\ell_{i,4}}$.
It is clear that $t^hM(0) + M_2(h,0) = O(n^{{\log_{t_0}{t}}}M^{1-\log_{r}{p_1}})$.
Then, it holds.
\end{proof}
\begin{cor}\label{cor4}
Let $m_0=n_0=k_0=t_0$. The IO complexity of Algorithm \ref{alg:W5} with input $A,B\in R^{n\times n}$ is no more than $$(\frac{r-1}{3r})^{-\frac{\ln{p_1}}{\ln{r}}}n^{{\log_{t_0}{t}}}M^{1-\frac{\ln{p_1}}{\ln{r}}} + 3\sum\limits_{i=1}^3\frac{\ell_{i,1}(n^{{\log_{t_0}{t}}} - n^{{\log_{t_0}{t-p_1+\ell_{i,3}}}})}{p_1 - \ell_{i,3}}(\frac{r-1}{3r})^{\frac{\ln{\ell_{i,3}}-\ln{{p_1}}}{\ln{r}}}M^{\frac{\ln{\ell_{i,3}}-\ln{p_1}}{\ln{r}}}$$
$$+3\sum\limits_{i=1}^3\frac{\ell_{i,2}}{t-p_1+\ell_{i,3}-t_0^2}({(\frac{M}{\beta})}^{1-\frac{\ln{t-p_1}}{2\ln{t_0}}}n^{\log_{t_0}{(t-p_1+\ell_{i,3}(\frac{t_0^2}{t-p_1})^{-\frac{\ln{r}}{2\ln{t_0}}})}}-n^2)$$
\end{cor}
\begin{proof}
By Theorem \ref{IOL2}, we have $M(h,0)\leq t^hM(0) + M_2(h,0) + M_3(h,0)$.
First, we have
$$t^hM(0)\leq (\frac{r-1}{3r})^{-\frac{\ln{p_1}}{\ln{r}}}n^{{\log_{t_0}{t}}}M^{1-\frac{\ln{p_1}}{\ln{r}}}$$
and
$$M_2(h,0)\leq3\sum\limits_{i=1}^3\frac{\ell_{i,1}(n^{{\log_{t_0}{t}}} - n^{{\log_{t_0}{t-p_1+\ell_{i,3}}}})}{p_1 - \ell_{i,3}}(\frac{r-1}{3r})^{\frac{\ln{\frac{\ell_{i,3}}{p_1}}}{\ln{r}}}M^{\frac{\ln{\ell_{i,3}}-\ln{p_1}}{\ln{r}}}$$

For $Z(0,x)<M$, we have
$$T_{x}(h,0) = 3\sum\limits_{i=1}^3\ell_{i,2}{h\choose {x}}(\frac{t_0^2}{t-p_1})^{h'}\frac{(t-p_1)^{h-x}}{t-p_1+\ell_{i,3}-t_0^2}(\ell_{i,3}(\frac{t_0^2}{t-p_1})^{-\frac{\ln{r}}{2\ln{t_0}}})^{x}, $$

For $Z(0,x)\geq M$, we have
$$T_{x}(h,0) \leq 3\sum\limits_{i=1}^3\ell_{i,2}{h\choose {x}}(\frac{t_0^2}{t-p_1})^{h'}\frac{(t-p_1)^{h-x}}{t-p_1+\ell_{i,3}-t_0^2}(\ell_{i,3}(\frac{t_0^2}{t-p_1})^{-\frac{\ln{r}}{2\ln{t_0}}})^{x}, $$

Then, we have
$$M_3(h,0)\leq 3\sum\limits_{i=1}^3\frac{\ell_{i,2}}{t-p_1+\ell_{i,3}-t_0^2}({(\frac{M}{\beta})}^{1-\frac{\ln{t-p_1}}{2\ln{t_0}}}n^{\log_{t_0}{(t-p_1+\ell_{i,3}(\frac{t_0^2}{t-p_1})^{-\frac{\ln{r}}{2\ln{t_0}}})}}-n^2)$$
Then, it holds.
\end{proof}
\section{Example}

\begin{lemma}\cite{1973Duality}\label{lle1}
Let $<U_{n_0\times m_0}, V_{m_0\times k_0}, W_{n_0\times k_0}>$ be an encoding/decoding matrix.
Then, $<U\otimes V\otimes W, V\otimes W \otimes U, W\otimes U \otimes V>$ is an encoding-decoding matrix($\otimes$ is the kronecker product).
\end{lemma}
Thus, we only need to give the basis transformations and algebra decomposition for $<U,V,W>$.
Now, we give an example of $<3,3,3;23>$-algorithm.

\textbf{Basis Transformations:}
In Appendix \textrm{I}, we give basis transformation matrices $\eta_1, \eta_2, \eta_3$ of the encoding/deconding matrices $<U_{23\times 9}, V_{23\times 9}, W_{23\times 9}>$.

\textbf{Algebra Decomposition:}
Let $U' = U\eta_1^{-1}, V' = V\eta_2^{-1}, W' = W{\eta_3}^{-1}$, $\psi'_1(A) = U'x, \psi'_2(B) = V'y, \psi'_3(C) = W'z, \phi'_1(B) = V'y, \phi'_2(C) = W'z, \phi'_3(A) = U'x, \varphi'_1(z) = W'^{T}z, \varphi'_2(x) = U'^{T}x, \varphi'_3(y) = V'^{T}y$ where $A = (a_{i,j})_{3\times 3}, B = (a_{i,j})_{3\times 3}, C = (c_{i,j})_{3\times 3}, x = (a_{i,j})_{1\times 9}, y = (b_{i,j})_{1\times 9}, z = (c_{i,j})_{1\times 9}$.

In Appendix \textrm{I}, we give the corresponding matrices of linear maps $\varphi'_{j,0}$ and $\ell_{j,i}', \phi_{j,i,1}', \phi_{j,i,2}',$ $\psi_{j,i,1}', \psi_{j,i,2}', \varphi'_{j,i}$ for $1\leq j\leq 3, 1\leq i\leq 22$ which satisfy that $\varphi'_{j} = {\sum\limits_{i=1}^{22}}^{\varphi'_{j,0}}\varphi'_{j,i}\ell'_{j,i}$ and
$\psi'_{j,i,2}\psi'_{j,i,1} = \ell'_{j,i}\psi'_j$, $\phi'_{j,i,2}\phi'_{j,i,1} = \ell'_{j,i}\phi'_j$, for $1\leq j\leq 3, 1\leq i\leq 22$, and specially, $\phi'_{j,i,2} = \psi'_{j,i,2} = \varphi'_{j,i} = I$ for $1\leq j\leq 3, 2\leq i\leq 22$.
Denote $\psi'_{j,i} = \ell'_{j,i}\psi'_j, \phi'_{j,i} = \ell'_{j,i}\phi'_j$ for $1\leq j\leq 3, 1\leq i\leq 22$.

\begin{algorithm}[t]\label{tmp}
\SetAlgoNoLine
\KwIn{S,T,level}
\eIf{$level==3$}{
   \Return{$ {\prod\limits_{i=1}^3}^{\circ}\varphi_{i,1}'({\prod\limits_{j=3}^1}^{\ast}\psi_{j,1,2}'(S)\odot{\prod\limits_{j=3}^1}^{\ast}\phi_{j,1,2}'(T))$}
}
{
    Let $R = \textbf{0}$\;
    \For{$i=2$ to $22$}{
        Compute $s' = \psi_{level+1,i,1}'(s)$ for $s\in S$\;
        Compute $t' = \psi_{level+1,i,1}'(t)$ for $t\in T$\;
        Denote $S' = (s'), T' = (t')$\;
        \If{$level\geq 1$}{
        Compute $S' = {\prod\limits_{j=level}^{1}}^{\ast}\psi_{j,1,2}'(S'), T' = {\prod\limits_{j=level}^{1}}^{\ast}\psi_{j,1,2}'(T')$\;}
        \If{$level\leq 1$}{
        Compute $c' = {\prod\limits_{k=level+2}^{3}}^{\circ}\varphi'_{k}({\prod\limits_{j=3}^{level}}^{\ast}\psi'_{j}(s')\odot{\prod\limits_{j=3}^{level}}^{\ast}\phi'_{j}(t'))$\;}
        \Else{
            $c' = s'\odot t'$\;
        }
        Denote $C' = (c')$\;
        \If{$level\geq 1$}{$C'= {\prod\limits_{j=1}^{level}}^{\circ}\varphi_{j,1}'(C)$\;}
        Compute $C'$ and add it to $R$ as $\varphi_{level+1,0}'$\;

     }
     Compute $s' = \psi_{level+1,1,1}'(s)$ for $s\in S$\;
     Compute $t' = \psi_{level+1,1,1}'(t)$ for $t\in T$\;
     Denote $S' = (s'), T' = (t')$\;
     Compute $\mathcal{A}(S',T', level+1)$ and add it to $R$ as $\varphi_{level+1,0}$\;
}
\Return{R}
\caption{$\mathcal{A}$}
\label{alg:W9}
\end{algorithm}

\textbf{Combine Algebra Decomposition:} We get an algebra decomposition of $<U'\otimes V'\otimes W', V'\otimes W' \otimes U', W'\otimes U' \otimes V'>$ by combining $\varphi'_{j,0}$ and $\ell_{j,i}', \phi_{j,i,1}', \phi_{j,i,2}', \psi_{j,i,1}', \psi_{j,i,2}', \varphi'_{j,i}$ for $1\leq j\leq 3, 1\leq i\leq 22$.

Let $\psi = \psi'_3 \ast \psi'_2 \ast \psi'_1, \phi = \phi'_3 \ast \phi'_2 \ast \phi'_1, \varphi = \varphi'_1\circ \varphi'_2\circ \varphi'_3$, $\psi_{i,2} = \phi_{i,2} = \varphi_{i} = I$ for $2\leq i\leq 12160$,
$\ell_{1} = \ell_{3,1}\ast\ell_{2,1}\ast\ell_{1,1}$, $\psi_{1,i} = \psi_{3,1,i}'\ast \psi_{2,1,i}'\ast\psi_{1,1,i}', \phi_{1,i} = \phi_{3,1,i}'\ast \phi_{2,1,i}'\ast\phi_{1,1,i}'$ for $1\leq i\leq 2$,
and $\psi_{i,1}, \phi_{i,1}$ be other components of $\psi,\phi$ for $2\leq i\leq 12160$ respectively, $\varphi_{0},\ell_{i}$ be the corresponding linear maps for $2\leq i\leq 12160$. We give them as following in detail:

Let $\mathcal{E}_i$ be an interception map with returning the $i$-th element,
$\ell_{4j+i-7} = \mathcal{E}_i\ast \ell_{2,1}\ast\ell_{1,1}, \psi_{4j+i-7,1} = \mathcal{E}_i\psi'_{3,j}\ast \psi'_{2,1}\ast\psi'_{1,1}, \phi_{i,1} = \phi'_{3,i}\ast \phi'_{2,1}\ast\phi'_{1,1}$ for $1\leq i\leq 4, 2\leq j\leq 22$,
$\ell_{46j+i-7} = \mathcal{E}_i\ast \ell_{2,j}\ast\ell_{1,1},\psi_{46j+i-7,1} = \mathcal{E}_i\psi'_3\ast \psi'_{2,j}\ast\psi'_{1,1}, \phi_{23j+i-24,1} = \mathcal{E}_i\phi'_3\ast \phi'_{2,j}\ast\phi'_{1,1}$ for $1\leq i\leq 46,2\leq j\leq 22$,
$\ell_{529j+i-7} = \mathcal{E}_i\ast \ell_{1,j}, \psi_{529j+i-7,1} = \mathcal{E}_i\psi'_3\ast \psi'_{2}\ast\psi'_{1,j}, \phi_{529j+i+505,1} = \mathcal{E}_i\phi'_3\ast \phi'_{2,j}\ast\phi'_{1,1}$ for $1\leq i\leq 529, 2\leq j\leq 22$,
and $c_1= d_1$, $c_i = \varphi'_1\circ \varphi'_2(\bigoplus\limits_{j=1}^4d_{i*4+j-7})$ for $2\leq i\leq 22$, $b_1 = {\sum\limits_{i=1}^{22}}^{\varphi_{3,0}}c_i$, $b_i = \varphi'_{1,1}(\bigoplus\limits_{j=1}^2\varphi'_{3}(\bigoplus\limits_{k=-29+23j+46i}^{-7+23j+46i}d_k))$ for $2\leq i\leq 22$,
$a_1 = {\sum\limits_{i=1}^{22}}^{\varphi'_{2,0}}b_i$,
$a_i = \varphi'_{2}\circ\varphi'_{3}(\bigoplus\limits_{j=-6+529*i}^{522+529*i}d_j)$ for $2\leq i\leq 22$,
${\sum\limits_{i=1}^{12160}}^{\varphi_0}d_i = {\sum\limits_{i=1}^{22}}^{\varphi_{1,0}'}a_i$.

Note that $\varphi = {\sum\limits_{i=1}^{12160}}^{\varphi_0}\varphi_i\ell_{i}$ and
$\psi_{i,2}\psi_{i,1} = \ell_{i}\psi$, $\phi_{i,2}\phi_{i,1} = \ell_{i}\phi$ for $1\leq i\leq 12160$(This is the assumption when we design the algorithm). Notice that giving $\varphi'_{j,0}$ and $\ell_{j,i}', \phi_{j,i,1}', \phi_{j,i,2}', \psi_{j,i,1}',$ $\psi_{j,i,2}', \varphi'_{j,i}$ for $1\leq j\leq 3$ is enough.
Then, in Appendix \textrm{II}, we present the $\varphi'_{j,0}$ and $\ell_{j,i}', \phi_{j,i,1}', \phi_{j,i,2}', \psi_{j,i,1}', \psi_{j,i,2}', \varphi'_{j,i}$ matrices with $1\leq j\leq 3, 1\leq i\leq t-p_1+1$ for other $<n,m,k;t>$-algorithms.

\textbf{Algorithm} $\mathcal{A}$: We give Algorithm $\mathcal{A}$(Algorithm 7)
to compute $\mathcal{F}(\varphi_{1}(\psi_{1,2}\odot\phi_{1,2}),\odot,A',B')$ where $\mathcal{A}((A'), (B'), 0)$ equals $\mathcal{F}(\varphi_{1}(\psi_{1,2}\odot\phi_{1,2}),\odot,A',B')$.
Notice that $\mathcal{A}$ depends on $\psi'_{i}, \phi'_{i}, \varphi'_{i}, $ $\varphi'_{i,0},$ $\psi'_{i,1,j}, \phi'_{i,1,j}$ for $q\leq i\leq 3, 1\leq j\leq 2$ and $\psi'_{i,j,1}, \phi'_{i,j,1}$ for $1\leq i\leq 3, $ $1\leq j\leq t-p_1+1$.
Then, we don't provide $\mathcal{A}$ in Appendix \textrm{II} for other $<n,m,k;t>$-algorithms.

\textbf{Needed Memory of Algorithm} $\mathcal{A}$: Assume that $\psi'_{i,1,1}(x), \phi'_{i,1,1}(x), \varphi'_{i,1,1}$ gets $\lambda_{i,1},\lambda_{i,2},$ $\lambda_{i,3}$ length vector for $1\leq i\leq 3$ respectively.
Let $u_1 = n_0m_0, u_2 = m_0k_0, u_3 = n_0k_0$.
Notice that lines 4-22 can be computed in $\sum\limits_{i=1}^3(u_1u_2u_2z_i+\frac{u_1u_2u_3}{u_i}z_i+z_i) + u_3z_1+u_1z_2+u_2z_3$ memory size when $level = 0$, Line 24 needs additional $\lambda_{1,1}u_2u_3z_1$ memory size when $level = 0$.
Since $S$ will not be used after Line 24, line 25 needs additional $\max\{0,\lambda_{1,2}u_1u_3z_2 - u_1u_2u_3z_1\}$ memory size when $level = 0$.
Similarly with this, we can get the needed memory size.
We give it in detail as following:

Let $\lambda_{0,i}=1 , e_{1,i} = \frac{u_1u_2u_3}{u_i}z_i, e_{3,z} = z_i$ for $1\leq i\leq 3, 1+\sum\limits_{x=1}^{i-1}\prod\limits_{y=0}^{2}\lambda_{y,x}\leq z\leq \sum\limits_{x=1}^i\prod\limits_{y=0}^{2}\lambda_{y,x}$,
$\mu_{i,z} = \lambda_{i,j}$ for $1\leq i\leq 3, 1\leq j\leq 3,  1+\sum\limits_{x=1}^{j-1}\prod\limits_{y=0}^{i-1}\lambda_{y,x}\leq z\leq \sum\limits_{x=1}^{j}\prod\limits_{y=0}^{i-1}\lambda_{y,x}$,
$e_{0,j} = u_1u_2u_3z_j$ for $1\leq j\leq 3$,
$e_{2,z} = u_3z_1$ for $1\leq z\leq \lambda_{1,1}$,
$e_{2,z} = u_1z_2$ for $1+\lambda_{1,1}\leq z\leq \lambda_{1,2}$,
$e_{2,z} = u_2z_3$ for $1+\lambda_{1,2}\leq z\leq \lambda_{1,3}$,
$\delta_{1,0} = \sum\limits_{i=1}^3u_1u_2u_3z_i, \Delta_{1,0} = u_1u_2u_3z_i$,
$\delta_{i,j} = \delta_{i,j-1}+\max\{0,\mu_{i,j-1}e_{i,j}-\delta_{i,j-1}+\Delta_{i,j-1}\}$ for $1\leq j\leq \sum\limits_{x=1}^3\prod\limits_{y=0}^{i-1}\lambda_{y,x}, 1\leq i\leq 3$,
$\Delta_{i,j} = \Delta_{i,j-1}+\mu_{i,j}e_{i,j} - e_{i-1,j}$ for $1\leq i\leq 3, 1\leq j\leq \sum\limits_{x=1}^2\prod\limits_{y=0}^{i-1}\lambda_{y,x}$, $\Delta_{i,j} = \Delta_{i,j}+\mu_{i,j}e_{i,j}$ for $1\leq i\leq 3, 1+\sum\limits_{x=1}^2\prod\limits_{y=0}^{i-1}\lambda_{y,x}\leq j\leq \sum\limits_{x=1}^3\prod\limits_{y=0}^{i-1}\lambda_{y,x}$, $\delta_{i,0} = \delta_{i-1,j}, \Delta_{i,0} = \Delta_{i-1,j}$ for $2\leq i\leq 4$ where $j = \sum\limits_{x=1}^3\prod\limits_{y=0}^{i-1}\lambda_{y,x}$,
$\delta_4 = \delta_{4,0} + \max\{0,\sum\limits_{i=1}^3(\lambda_{2,i}\lambda_{3,i}z_i+\lambda_{3,i}z_i+z_i)-\delta_{4,0}+\Delta_{4,0}\}$,
$\delta_3 = \max\{\delta_{4,0}, \delta_4,\sum\limits_{i=1}^3z_i+ \sum\limits_{i=1}^3\lambda_{2,i}z_i + \sum\limits_{i=1}^3\lambda_{1,i}\lambda_{2,i}z_i-\delta_{3,0}+\Delta_{3,0}\}$, $\delta_2 = \max\{\delta_{3,0},\delta_3, (u_3z_1+u_1z_2+u_2z_3)+\lambda_{1,1}u_3z_1+\lambda_{1,2}u_1z_2+\lambda_{1,3}u_2z_3 + \sum\limits_{i=1}^3z_i -\delta_{2,0}+\Delta_{2,0}\}$,
$\delta_1 = \max\{\delta_{2,0},\delta_2,\sum\limits_{i=1}^3(u_1u_2u_2z_i+\frac{u_1u_2u_3}{u_i}z_i+z_i) + u_3z_1+u_1z_2+u_2z_3\}$.
Then Algorithm $\mathcal{A}$ can be computed within memory $\delta_1$.
Let $\delta_1 = \sum\limits_{i=1}^3\beta_iz_i$, $\delta'_1$ be the value with setting $z_1 = z_2 = z_3 = 1$ in $\delta_1$. Then, $\delta'_1\geq \sum\limits_{i=1}^3\beta_i$.

\textbf{IO complexity}:
By Corollary \ref{cor4}, Observation \ref{ok1}, and Lemma \ref{lemmak1}, the IO complexity of Algorithm \ref{fin1} is
$14n^{\log_323}M^{-0.5} - 6n^{\frac{\log_{27}(23^3-4)}{3}}M^{-0.5} + 28.07n^{\frac{\log_{27}{23^3-0.77}}{3}}M^{-0.42} - 16.69n^2 + 2.33n^2\log_{3}{\sqrt{2}\frac{n}{
\sqrt{M}}}+2M$.
And by Theorem \ref{IOA3}, Observation \ref{ok1}, and Lemma \ref{lemmak1},
the IO complexity of Algorithm \ref{fin2} is $10.08n^{\log_{3}23}M^{-0.42} + 23.02n^{\log_{27}(23^3-4)}M^{-0.42} - 16.69n^2 + 2.33n^2\log_{3}{\sqrt{2}\frac{n}{
\sqrt{M}}}+2M$.

\textbf{Arithmetic complexity}: By Theorem \ref{a2}, Observation \ref{ok1}, and Lemma \ref{lemmak1},
the arithmetic complexity of Algorithm \ref{fin1} is $2n^{\log_323} + 4.56n^{\log_{27}(23^3-4)} - 5.56n^2 + 0.77n^2\log_3n$.
By Theorem \ref{a4}, Observation \ref{ok1}, and Lemma \ref{lemmak1}, the arithmetic complexity of Algorithm  \ref{fin2} is
$(2+4.56M^{\frac{\ln{(23^3-4)}-\ln{23^3}}{6\ln3}} - 4.58M^{\frac{2\ln{3}-\ln{23}}{2\ln{3}}})n^{\log_323} + 5.56 n^2 +$ $ 0.77n^2\log_3n+$ $ (9.35M^{1-\frac{\ln(23^3-4)}{6\ln{3}}} - $ $ M^{\frac{\ln{(23^3-4)}-3\ln{23}}{6\ln{3}}})n^{\log_{27}(23^3-4)}$.

\section{Acknowledgments}
This work was supported by 
National key research and development program of China(Grant No.2019YFA0706401) and 
the National Natural Science Foundation of China (Grant No.62172116, No.62172014, No.62172015, No.61872166, No.62002002(Youth Foundation)). 
\bibliographystyle{plain}
\bibliography{sample-base}

\begin{thebibliography}{10}

\bibitem{2013On}
Valerii~Borisovich Alekseev and A.~V. Smirnov.
\newblock On the exact and approximate bilinear complexities of multiplication
  of $4\times 2$ and $2\times 2$ matrices.
\newblock {\em Proceedings of the Steklov Institute of Mathematics},
  282(1):123--139, 2013.

\bibitem{ballard2013graph}
Grey Ballard, James Demmel, Olga Holtz, and Oded Schwartz.
\newblock Graph expansion and communication costs of fast matrix
  multiplication.
\newblock {\em Journal of the ACM (JACM)}, 59(6):1--23, 2013.

\bibitem{beniamini2019faster}
Gal Beniamini and Oded Schwartz.
\newblock Faster matrix multiplication via sparse decomposition.
\newblock In {\em The 31st ACM Symposium on Parallelism in Algorithms and
  Architectures}, pages 11--22, 2019.

\bibitem{benson2015framework}
Austin~R Benson and Grey Ballard.
\newblock A framework for practical parallel fast matrix multiplication.
\newblock {\em ACM SIGPLAN Notices}, 50(8):42--53, 2015.

\bibitem{2012O}
Dario Bini, Milvio Capovani, Francesco Romani, and Grazia Lotti.
\newblock {O}($n^{2.7799}$) complexity for $n\times n$ approximate matrix
  multiplication.
\newblock {\em Information Processing Letters}, 8(5):234--235, 1979.

\bibitem{cenk2017arithmetic}
Murat Cenk and M~Anwar Hasan.
\newblock On the arithmetic complexity of strassen-like matrix multiplications.
\newblock {\em Journal of Symbolic Computation}, 80:484--501, 2017.

\bibitem{coppersmith1982asymptotic}
Don Coppersmith and Shmuel Winograd.
\newblock On the asymptotic complexity of matrix multiplication.
\newblock {\em SIAM Journal on Computing}, 11(3):472--492, 1982.

\bibitem{coppersmith1987matrix}
Don Coppersmith and Shmuel Winograd.
\newblock Matrix multiplication via arithmetic progressions.
\newblock In {\em Proceedings of the 19th annual ACM symposium on Theory of
  computing}, pages 1--6, 1987.

\bibitem{2004Getting}
Susan~L. Graham, Marc Snir, and Cynthia~A. Patterson.
\newblock Getting up to speed: The future of supercomputing.
\newblock {\em National Academies Press Washington Dc}, 149(1):147–153, 2004.

\bibitem{1973Duality}
John. Hopcroft and J.~Musinski.
\newblock Duality applied to the complexity of matrix multiplication and other
  bilinear forms.
\newblock {\em SIAM Journal on Computing}, 2(3):159–173, 1973.

\bibitem{karstadt2020matrix}
Elaye Karstadt and Oded Schwartz.
\newblock Matrix multiplication, a little faster.
\newblock {\em Journal of the ACM (JACM)}, 67(1):1--31, 2020.

\bibitem{le2014powers}
Fran{\c{c}}ois Le~Gall.
\newblock Powers of tensors and fast matrix multiplication.
\newblock In {\em Proceedings of the 39th international symposium on symbolic
  and algebraic computation}, pages 296--303, 2014.

\bibitem{pan1978strassen}
Victor~Ya Pan.
\newblock Strassen's algorithm is not optimal trilinear technique of
  aggregating, uniting and canceling for constructing fast algorithms for
  matrix operations.
\newblock In {\em 19th Annual Symposium on Foundations of Computer Science
  (sfcs 1978)}, pages 166--176. IEEE, 1978.

\bibitem{probert1976additive}
Robert~L Probert.
\newblock On the additive complexity of matrix multiplication.
\newblock {\em SIAM Journal on Computing}, 5(2):187--203, 1976.

\bibitem{Scho1981PARTIAL}
A.~Sch\"{o}nhage.
\newblock Partial and total matrix multiplication.
\newblock {\em SIAM Journal on Computing}, 10(3):434--455, 1981.

\bibitem{smirnov2013bilinear}
Alexey~V Smirnov.
\newblock The bilinear complexity and practical algorithms for matrix
  multiplication.
\newblock {\em Computational Mathematics and Mathematical Physics},
  53(12):1781--1795, 2013.

\bibitem{strassen1969gaussian}
Volker Strassen.
\newblock Gaussian elimination is not optimal.
\newblock {\em Numerische mathematik}, 13(4):354--356, 1969.

\bibitem{williams2012multiplying}
Virginia~Vassilevska Williams.
\newblock Multiplying matrices faster than coppersmith-winograd.
\newblock In {\em Proceedings of the 44th annual ACM symposium on Theory of
  computing}, pages 887--898, 2012.

\end{thebibliography}

\break
\section{Appendix}
\textbf{Appendix I:}

$\eta_1 = \eta_2 = I$.
\begin{figure}[h]
  \centering
  \includegraphics[width=300pt]{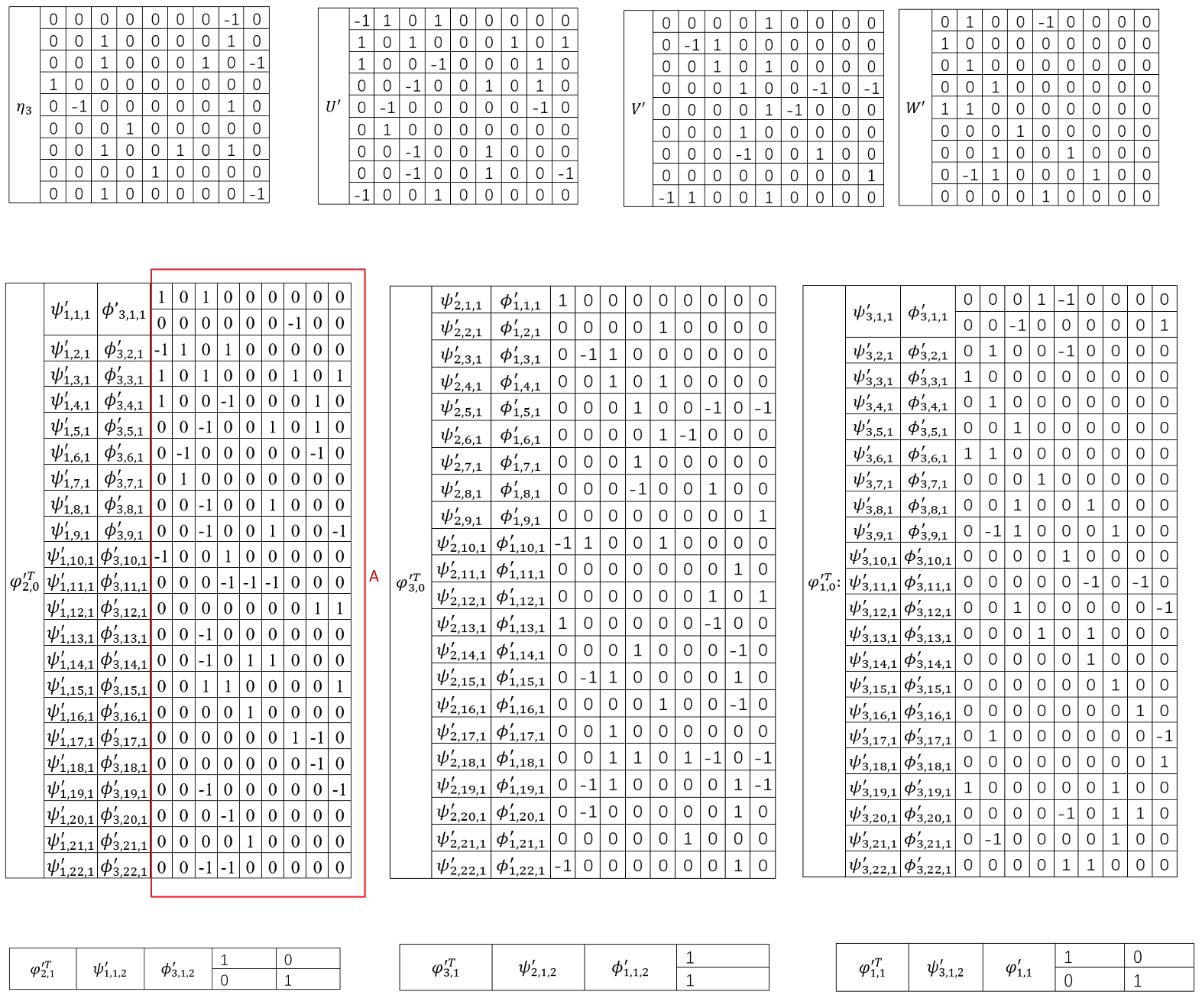}
  \caption{Algebra Decomposition of $<3,3,3;23>$. Note :$\varphi'_{2,0}(z_1,z_2,...,z_{22}) := A^{T}\bigoplus\limits_{i=1}^{22}z_i$. The definitions of other linear maps are similar. }
\end{figure}

All the data of Appendix I and Append II can be found at https://github.com/wp-hhh/Algebra-decomposition-Algorithm.

\end{document}